\documentclass[11pt]{article}

\usepackage{amsmath}
\usepackage{amsthm}
\usepackage{amssymb}
\usepackage{graphics}
\usepackage{url}

\usepackage{amsfonts}
\usepackage{hhline}
\usepackage{multirow}
\usepackage{xcolor}
\usepackage{float}
\usepackage{bbm}
\usepackage{algorithm}
\usepackage{algpseudocode}
\usepackage{comment}
\usepackage{caption}
\usepackage{subcaption}
\usepackage{pgf}
\usepackage{tikz}
\usepackage{enumitem}
\usepackage[left=1in,right=1in,top=1in,bottom=1in]{geometry}
\usepackage{nicefrac}

\usepackage{float}
\usepackage[font=bf,skip=\baselineskip]{caption}
\usepackage{wrapfig}
\usepackage{soul}
\usepackage{times}

\allowdisplaybreaks

\newcommand{\OPT}{\mathcal{OPT}}

\let\citep\cite

\newtheorem{theorem}{Theorem}
\newtheorem{lemma}[theorem]{Lemma}

\newtheorem{fact}[theorem]{Fact}

\newtheorem{definition}[theorem]{Definition}

\newenvironment{sproof}%
{%
 \par\noindent{\it Proof sketch. \ }%
}%
{\qed}

\title{Clustering under Local Stability: Bridging the Gap between Worst-Case and Beyond Worst-Case Analysis
\footnote{Authors' addresses: 
\texttt{ninamf@cs.cmu.edu, crwhite@cs.cmu.edu}. 
This work was supported in part by grants NSF-CCF 1535967, NSF CCF-1422910, NSF IIS-1618714, a Sloan Fellowship, a Microsoft Research Fellowship, and a National Defense Science and Engineering Graduate (NDSEG) fellowship.}
}

\author{Maria-Florina Balcan \and
Colin White}
\date{}

\begin{document}

\maketitle

\begin{abstract}

Recently, there has been substantial interest in clustering research that takes
a beyond worst-case approach to the analysis of algorithms.
The typical idea is to design a clustering algorithm that outputs a near-optimal solution,
provided the data satisfy a natural stability notion.
For example, Bilu and Linial (2010) and Awasthi et al.\ (2012)
presented algorithms that output near-optimal solutions, assuming the optimal solution is preserved under small perturbations to the input distances.
A drawback to this approach is that the algorithms are often explicitly built according to the stability assumption
and give no guarantees in the worst case;
indeed, several recent algorithms output arbitrarily bad solutions even when just a small section of
the data does not satisfy the given stability notion.

In this work, we address this concern in two ways.
First, we provide algorithms that inherit the worst-case guarantees of 
clustering approximation algorithms, while simultaneously guaranteeing near-optimal solutions
when the data is stable.
Our algorithms are natural modifications to existing state-of-the-art approximation algorithms.
Second, we initiate the study of local stability, which is a property of a single optimal cluster rather than
an entire optimal solution.
We show our algorithms output all optimal clusters which satisfy stability locally.
Specifically, we achieve strong positive results in our local framework under recent stability notions including
metric perturbation resilience (Angelidakis et al.\ 2017) and
robust perturbation resilience (Balcan and Liang 2012) for the 
$k$-median, $k$-means, and symmetric/asymmetric $k$-center objectives.

\end{abstract}


\allowdisplaybreaks
\setcounter{page}{0}
\thispagestyle{empty}
\newpage

\section{Introduction} \label{sec:intro}

Clustering is a fundamental problem in combinatorial optimization with numerous real-life applications
in areas from bioinformatics to computer vision to text analysis and so on.
The underlying goal is to group a given set of points to maximize similarity inside a group and
minimize similarity among groups.
A common approach to clustering is to set up an objective function and then approximately find the
optimal solution according to the objective. 
Given a set of points $S$ and a distance metric $d$, common clustering objectives include finding
$k$ centers to minimize the sum of the distance, or squared distance, from each point to its closest center 
($k$-median and $k$-means, respectively),
or to minimize the maximum distance from a point to its closest center ($k$-center).
These popular objective functions are provably NP-hard to optimize~\cite{gonzalez1985clustering,jain2002new,lee2017improved},
so research has focused on finding approximation algorithms. 
This has attracted significant attention in the theoretical computer science community
\cite{arya2004local,byrka2015improved,charikar1999constant,outliers,chen2008constant,gonzalez1985clustering,makarychev2016bi}.

Traditionally, the theory of clustering (and more generally, the theory of algorithms) 
has focused on the analysis of worst-case instances.
While this approach has led to many elegant algorithms and lower bounds, 
it is often overly pessimistic of an algorithm's performance on ``typical'' instances or real world instances.
A rapidly developing line of work in the algorithms community, 
the so-called \emph{beyond worst-case analysis} of  algorithms (BWCA),
considers designing algorithms for instances that satisfy some natural structural properties. 
BWCA has given rise to many positive results~\cite{hardt2013beyond,kumar2004simple,tim},
especially for clustering problems~\cite{awasthi2010stability,awasthi2012center,as,kumar2010clustering}.
For example, the popular notion of $\alpha$-perturbation resilience, introduced by Bilu and Linial~\cite{bilu2012stable},
informally states that the optimal solution does not change when
the input distances are allowed to increase by up to a factor of $\alpha$.
This definition seeks to capture a phenomenon in practice: the optimal solution is often 
significantly better than all other solutions,
thereby the optimal solution does not change even when the input is slightly perturbed.

However, there are two potential downsides to this approach.
The first downside is that many of these algorithms aggressively exploit the given structural assumptions,
which can lead to contrived algorithms with no guarantees in the worst case.
Therefore, a user can only use
algorithms from this line of work if she is certain her data satisfies the assumption
(even though none of these assumptions are computationally efficient to verify).
The second downside is that 
while the algorithms return the optimal solution when the input is stable, 
there are no partial guarantees when most, but not all, of the data is stable.
For example, the algorithms of Balcan and Liang~\cite{balcan2012clustering} and Angelidakis et al.~\cite{angelidakis2017algorithms}
return the optimal clustering when the instance is resilient to perturbations, however,
both algorithms use a dynamic programming subroutine that is susceptible to errors which can propagate when a small fraction
of the data does not satisfy perturbation resilience (see Appendix~\ref{app:algos} for more details).
From these setbacks, two natural questions arise. 
\emph{(1)} Can we find \emph{natural} algorithms that achieve the worst-case approximation ratio, while outputting the
optimal solution if the data is stable,
\footnote{ A trivial solution is to run an approximation algorithm and a BWCA algorithm in parallel, and output the better of
the two solutions. We seek a more satisfactory and ``natural'' answer to this question.
} and
\emph{(2)} Can we construct robust algorithms that still output good results even when only part of the data satisfies stability?
The current work seeks to answer both of these questions. 

\subsection{Our results and techniques}
In this work, we answer both questions affirmatively for a variety of clustering objectives under
perturbation resilience.
We present algorithms that simultaneously achieve state-of-the-art approximation ratios in the worst case,
while outputting the optimal solution when the data is stable. All of our algorithms are natural
modifications to existing approximation algorithms. 
To answer question \emph{(2)}, 
we define the notion of \emph{local perturbation resilience} in Section~\ref{sec:prelim}.
This is the first definition that applies to an individual cluster, rather than the dataset as a whole.
Informally, an optimal cluster satisfies $\alpha$-local perturbation resilience if it remains in the optimal solution
under any $\alpha$ perturbation to the input.
We show that every optimal cluster satisfying local stability will be returned by our algorithms.
Therefore, given an instance with a mix of stable and non-stable data, our algorithms will return the optimal clusters
over the stable data, and a worst-case approximation guarantee over the rest of the data.
Specifically, we prove the following results. 

\paragraph{Approximation algorithms under local perturbation resilience}
In Section~\ref{sec:approx}, 
we introduce a condition that is sufficient for an $\alpha$-approximation algorithm to return the optimal $k$-median
or $k$-means clusters satisfying $\alpha$-local perturbation resilience.
Intuitively, our condition requires that the approximation guarantee must be true locally as well as globally.
We show the popular local search algorithm satisfies the property for a sufficiently large search parameter.
For $k$-center, we show that \emph{any} $\alpha$-approximation algorithm for $k$-center will always return the
clusters satisfying $\alpha$-local metric perturbation resilience, which is a slightly weaker notion of perturbation resilience.

\paragraph{Asymmetric $k$-center}
The asymmetric $k$-center problem admits an $O(\log^* n)$ approximation algorithm due to Vishnwanathan~\cite{vishwanathan},
which is tight~\cite{chuzhoy2005asymmetric}. In Section~\ref{sec:akc},
we show a simple modification to this algorithm ensures that it returns all optimal clusters satisfying a condition slightly
stronger than 2-local perturbation resilience (all neighboring clusters must satisfy 2-local perturbation resilience).
The first phase of Vishwanathan's approximation algorithm involves iteratively removing the neighborhood of special points called
\emph{center-capturing vertices (CCVs)}. We show the centers of locally perturbation resilient clusters are CCVs and satisfy a separation condition,
by constructing 2-perturbations in which neighboring clusters cannot be too close to the 
locally perturbation resilient centers without causing a contradiction.
This allows us to modify the approximation algorithm by first removing the neighborhood around CCVs satisfying the separation condition.
With a careful reasoning, we maintain the original approximation guarantee while only removing points from a single local perturbation resilient
cluster at a time.

\paragraph{Robust perturbation resilience}
In Section~\ref{sec:3eps}, we consider $(\alpha,\epsilon)$-local perturbation resilience, 
which states that at most $\epsilon n$ points can swap into or out of the cluster under any $\alpha$-perturbation.
For $k$-center, we show that any 2-approximation algorithm will return
the optimal $(3,\epsilon)$-locally perturbation resilient clusters, assuming a mild lower bound on optimal cluster sizes.
To prove this, we show that if points from two different locally perturbation resilient clusters are close to each other, 
then $k-1$ centers achieve the optimal value under a carefully constructed 3-perturbation. 
The rest of the analysis involves building up conditional claims
dictating the possible centers for each locally perturbation resilient cluster under the 3-perturbation.
We utilize the idea of a \emph{cluster-capturing center}~\cite{symmetric} along with machinery specific to handling local perturbation
resilience to show that a locally perturbation resilient cluster must split into two clusters under the 3-perturbation, causing a contradiction.
Finally, we show that the mild lower bound on the cluster sizes is necessary. Specifically, we show
hardness of approximation for $k$-median, $k$-means, and $k$-center, 
even when it is guaranteed the clustering satisfies $(\alpha,\epsilon)$-perturbation resilience for any $\alpha\geq 1$ and $\epsilon>0$.
In fact, the result holds even for a stronger notion called $(\alpha,\epsilon)$-approximation stability.
The hardness is based on a reduction from the general clustering
instances, so the APX-hardness constants match the worst-case APX-hardness results
of 2, 1.73, and 1.0013 for $k$-center~\cite{gonzalez1985clustering}, $k$-median~\cite{jain2002new}, 
and $k$-means~\cite{lee2017improved}, respectively.
This generalizes prior hardness results in BWCA~\cite{as,symmetric}.

\subsection{Related work} \label{sec:related_work}

\noindent {\bf Clustering}~~
The first constant-factor approximation algorithm for $k$-median was given by
Charikar et al.~\cite{charikar1999constant},
and the current best approximation ratio is 2.675 by Byrka et al.~\cite{byrka2015improved}.
Jain et al.\ proved $k$-median is NP-hard to approximate to a factor better than 1.73~\cite{jain2002new}.
For $k$-center, Gonzalez showed a tight 2-approximation algorithm \cite{gonzalez1985clustering}.
For $k$-means, the best approximation ratio was recently lowered to $6.357$ 
by Ahmadian et al.~\cite{ahmadian2016better}.
$k$-means was shown to be APX-hard by Awasthi et al.~\cite{awasthi2015hardness}, and the constant was recently improved to 1.0013~\cite{lee2017improved}.

\noindent {\bf Perturbation resilience}~~
Perturbation resilience was introduced by Bilu and Linial, who showed algorithms that outputted
the optimal solution for max cut under $\Omega(\sqrt{n})$-perturbation resilience~\cite{bilu2012stable}.
This result was improved by Markarychev et al.~\cite{makarychev2014bilu}, who
showed the standard SDP relaxation is integral for $\Omega(\sqrt{\log n}\log{\log{n}})$-perturbation resilient
instances.
They also show an optimal algorithm for minimum multiway cut under 4-perturbation resilience.
The study of clustering under perturbation resilience was initiated by Awasthi et al.~\cite{awasthi2012center},
who provided an optimal algorithm for center-based clustering objectives (which includes $k$-median, $k$-means, and $k$-center clustering,
as well as other objectives) under 3-perturbation resilience.
This result was improved by Balcan and Liang~\cite{balcan2012clustering}, who showed an algorithm for center-based clustering under $(1+\sqrt{2})$-perturbation resilience. 
They also gave a near-optimal algorithm for $k$-median $(2+\sqrt{3},\epsilon)$-perturbation resilience, a robust version of perturbation resilience, when the optimal clusters are not too small.
Balcan et al.~\cite{symmetric} constructed algorithms for $k$-center and asymmetric $k$-center under 2-perturbation resilience and $(3,\epsilon)$-perturbation resilience, and they showed no polynomial-time
algorithm can solve $k$-center under $(2-\epsilon)$-approximation stability (a notion that is stronger than perturbation resilience) unless $NP=RP$.
Recently, Angelidakis et al.~\cite{angelidakis2017algorithms},
gave algorithms for center-based clustering under 2-perturbation resilience
and minimum multiway cut with $k$ terminals 
under $(2-2/k)$-perturbation resilience.
They also define the more general notion of metric perturbation resilience.
In Appendix~\ref{app:algos}, we discuss prior work in the context of local
perturbation resilience.
Perturbation resilience has also been applied to other problems, such as the traveling salesman problem, and finding Nash 
equilibria~\cite{as,mihalak2011complexity}.

\paragraph{Approximation stability}
Approximation stability is a related definition that is stronger than
perturbation resilience.
It was introduced by Balcan et al.~\cite{as},
who showed algorithms that outputted nearly optimal solutions under $(\alpha,\epsilon)$-approximation stability
for $k$-median and $k$-means when $\alpha>1$.
Balcan et al.~\cite{balcan2009agnostic} studied a relaxed notion of approximation stability in which a specified $\nu$ fraction
of the data satisfies approximation stability. In this setting, there may not be a unique approximation stable
solution. The authors provided an algorithm which outputted a small list of clusterings, such that all
approximation stable clusterings are close to one clustering in the list.
We remark the property itself is similar in spirit to local stability, although the solution/results are much different.
Voevodski et al.~\cite{voevodski2011min} gave an algorithm for empirically clustering protein sequences
using the min-sum objective under approximation stability, 
which compares favorably to popular clustering algorithms used in practice.
Gupta et al.~\cite{gupta2014decompositions} showed algorithm for finding near-optimal solutions for $k$-median under approximation stability in the context of finding triangle-dense graphs.

\paragraph{Other stability notions}
Ostrovsky et al.\ show how to efficiently cluster instances
in which the $k$-means clustering cost is much lower than the $(k-1)$-means cost~\cite{ostrovsky2012effectiveness}.
Kumar and Kannan give an efficient clustering algorithm for instances in which
the projection of any point onto the line between its cluster center to any other cluster center
is a large additive factor closer to its own center than the other center~\cite{kumar2010clustering}.
This result was later improved along multiple axes by Awasthi and Sheffet~\cite{Awasthi2012Improved}.
There are many other works that show positive results for different natural notions of stability in various 
settings~\cite{arora2012,awasthi2010stability,gupta2014decompositions,hardt2013beyond,kumar2010clustering,kumar2004simple,tim}.

\section{Preliminaries}\label{sec:prelim}

A clustering instance consists of a set $S$ of $n$ points, as well as a distance function
$d:S\times S\rightarrow\mathbb{R}_{\geq 0}$.
For a point $u\in S$ and a set $A\subseteq S$, we define $d(u,A)=\min_{v\in A}d(u,v)$.
The $k$-median, $k$-means, and $k$-center objectives are to find a set of points
$X= \{x_1, \dots, x_k\}\subseteq S$ called \emph{centers}
to minimize $\sum_{v\in S} d(v,X)$, $\sum_{v\in S} d(v,X)^2$,  and $\max_{v\in S}d(v,X)$, respectively.
We denote by $\text{Vor}_{X}(x)$ the Voronoi tile of $x$ induced by $X$ on the set of points $S$,
and we denote $\text{Vor}_X(X')=\bigcup_{x\in X'}\text{Vor}_X(x)$ for a subset $X'\subseteq X$.
We refer to the Voronoi partition induced by $X$ as a clustering.
Throughout the paper, we denote the clustering  with minimum cost by $\OPT=\{C_1, \dots, C_k \}$, 
and we denote the optimal centers by $c_1,\dots,c_k$, where $c_i$ is the center of $C_i$ for all $1\leq i\leq k$.

All of the distance functions we study are metrics, except for Section~\ref{sec:akc}, in which we study an \emph{asymmetric} distance function.
An asymmetric distance function satisfies all the properties of a metric space except for symmetry.
In particular, an asymmetric distance function must satisfy the \emph{directed} triangle inequality: for all $u,v,w\in S$,
$d(u,w)\leq d(u,v)+d(v,w)$.

We formally define \emph{perturbation resilience}, a notion introduced by Bilu and Linial~\cite{bilu2012stable}.
$d'$ is called an $\alpha$-perturbation of the distance function $d$, if for all $u,v\in S$, $d(u,v)\leq d'(u,v) \leq \alpha d(u,v)$.
(We only consider perturbations in which the distances increase because WLOG we can scale
the distances to simulate decreasing distances.)

\begin{definition} \label{def:pr}
A clustering instance $(S,d)$ satisfies \emph{$\alpha$-perturbation resilience} ($\alpha$-PR) if for any $\alpha$-perturbation $d'$ of $d$,
the optimal clustering under $d'$ is unique and  equal to $\OPT$.
\end{definition}

Now we define \emph{local perturbation resilience}, a property of an optimal cluster rather than a dataset.

\begin{definition}\label{def:lpr}
Given a clustering instance $(S,d)$ with optimal clustering $\mathcal{C}=\{C_1,\dots,C_k\}$, 
an optimal cluster $C_i$ satisfies \emph{$\alpha$-local perturbation resilience} ($\alpha$-LPR) if for any $\alpha$-perturbation $d'$ of $d$,
the optimal clustering $\mathcal{C}'$ under $d'$ contains $C_i$.
\end{definition}

We will sometimes refer to a center $c_i$ of an $\alpha$-LPR cluster $C_i$ as an $\alpha$-LPR center.
Clearly, if a clustering instance is perturbation resilient, then every optimal cluster satisfies local perturbation resilience.
Now we will show the converse is also true.

\begin{fact} \label{fact:local-iff}
A clustering instance $(S,d)$ satisfies $\alpha$-PR if and only if
each optimal cluster satisfies $\alpha$-LPR.
\end{fact}

\begin{proof}
Given a clustering instance $(S,d)$, the forward direction follows by definition:
assume $(S,d)$ satisfies $\alpha$-PR, and given an optimal cluster $C_i$, then for each $\alpha$-perturbation $d'$, the optimal clustering
stays the same under $d'$, therefore $C_i$ is contained in the optimal clustering under $d'$.
Now we prove the reverse direction. Given a clustering instance with optimal clustering $\mathcal{C}$, and given an $\alpha$-perturbation
$d'$, let the optimal clusetring under $d'$ be $\mathcal{C}'$. For each $C_i\in\mathcal{C}$, by assumption, $C_i$ satisfies $\alpha$-LPR,
so $C_i\in\mathcal{C}'$. Therefore $\mathcal{C}=\mathcal{C}'$.
\end{proof}

In Section~\ref{sec:akc}, we define a stronger version of Definition~\ref{def:lpr} specifically for $k$-center.
Next, we define a more robust version of $\alpha$-PR and $\alpha$-LPR that allows a small change in the optimal clusters when the distances are perturbed.
We say that two clusters $A$ and $B$ are $\epsilon$-close if they differ by only $\epsilon n$ points, i.e., $|A\setminus B|+|B\setminus A|\leq\epsilon n$.
We say that two clusterings $\mathcal{C}$ and $\mathcal{C}'$ are $\epsilon$-close if
$\min_\sigma\sum_{i=1}^k |C_i\setminus C_{\sigma(i)}'|\leq\epsilon n$.

\begin{definition}~\cite{balcan2012clustering} \label{def:alpha-epsilon}
A clustering instance $(S,d)$ satisfies \emph{$(\alpha,\epsilon)$-perturbation resilience $((\alpha,\epsilon)$-PR)} if for any $\alpha$- perturbation $d'$ of $d$,
all optimal clusterings under $d'$ must be $\epsilon$-close to $\OPT$.
\end{definition}

\begin{definition}
Given a clustering instance $(S,d)$ with optimal clustering $\mathcal{C}=\{C_1,\dots,C_k\}$, 
an optimal cluster $C_i$ satisfies \emph{$(\alpha,\epsilon)$-local perturbation resilience} $((\alpha,\epsilon)$-LPR$)$ if for any $\alpha$-perturbation $d'$ of $d$,
the optimal clustering $\mathcal{C}'$ under $d'$ contains a cluster $C_i'$ which is $\epsilon$-close to $C_i$.
\end{definition}

We prove a statement similar to Fact~\ref{fact:local-iff}, for $(\alpha,\epsilon)$-PR, but the $\epsilon$ error adds up among the clusters.
See Appendix~\ref{app:prelim} for the proof.

\begin{lemma} \label{lem:eps-iff}
A clustering instance $(S,d)$ satisfies $(\alpha,\epsilon)$-PR if and only if
each optimal cluster $C_i$ satisfies $(\alpha,\epsilon_i)$-LPR and $\sum_i\epsilon_i\leq2\epsilon n$.
\end{lemma}

In all definitions thus far, we do not assume that the $\alpha$-perturbations satisfy the triangle inequality.
Angelidakis et al.~\cite{angelidakis2017algorithms} recently studied the weaker definition in which the $\alpha$-perturbations 
must satisfy the triangle inequality, called \emph{metric perturbation resilience}.
All of our definitions can be generalized accordingly, and some of our results hold under this weaker assumption.
To this end, we will sometimes take the \emph{metric completion} $d'$ of a non-metric distance function $d''$, 
by setting the distances in $d'$ as the length of the shortest path on the graph whose edges are the lengths in $d''$.

\section{Approximation algorithms under local perturbation resilience} \label{sec:approx}

In this section, we show that local search for $k$-median will always return the $(3+\epsilon)$-LPR clusters,
and for $k$-means it will return the $(9+\epsilon)$-LPR clusters.
We also show that \emph{any} 2-approximation for $k$-center will return the 2-LPR clusters.

\paragraph{$k$-median}

We start by giving a condition on an approximate $k$-median solution, which is sufficient to show
the solution contains all $\alpha$-LPR clusters.

\begin{lemma} \label{lem:approx}
Given a $k$-median instance $(S,d)$ and a set of $k$ centers $X$, if for all sets $Y$ of size $k$,
\begin{equation*}
\sum_{v\in \text{Vor}_X(X\setminus Y)\cup\text{Vor}_Y(Y\setminus X)} d(v,X)
\leq\sum_{v\in\text{Vor}_Y(Y\setminus X)}\min(d(v,X),\alpha d(v,Y))+ \alpha \sum_{v\in \text{Vor}_X(X\setminus Y)} d(v,Y),
\end{equation*}
then all $\alpha$-LPR clusters $C_i$ are contained in the clustering defined by $X$.
\end{lemma}

\begin{proof}
Given such a set of centers $X$, we construct an $\alpha$-perturbation $d'$ as follows. 
Increase all distances by a factor of $\alpha$ except for the distances between each point
$v$ and its closest center in $X$.
Now our goal is to show that $X$ is the optimal set of centers under $d'$.

Given any other set $Y$ of $k$ centers, we consider four types of points:
$\text{Vor}_X(X\cap Y)\cap\text{Vor}_Y(X\cap Y)$, 
$\text{Vor}_X(X\cap Y)\setminus\text{Vor}_Y(X\cap Y)$,
$\text{Vor}_Y(X\cap Y)\setminus\text{Vor}_X(X\cap Y)$, and 
$\text{Vor}_X(X\setminus Y)\cap \text{Vor}_Y(Y\setminus X)$,
which we denote by $A_1$, $A_2$, $A_3$, and $A_4$, respectively
(see Figure~\ref{fig:approx}).
The distance from a point $v\in S$ to its center in $Y$ might stay the same under $d'$, or increase, depending on its type.
For each point $v\in A_1$, $d'(v,Y)=d(v,Y)=d(v,X)$ because these points have centers in $X\cap Y$.
For each point $v\in A_3\cup A_4$, $d'(v,Y)=\alpha d(v,Y)$ because their centers are in $Y\setminus X$.
The points in $A_2$ were originally closest to a center in $Y\setminus X$, but might switch to their center in $X$, since it is in $X\cap Y$.
Therefore, for each $v\in A_2$, $d'(v,Y)= \min(d(v,X),\alpha d(v,Y))$. Altogether,
\begin{align*}
\sum_{v\in S}d'(v,Y)\geq\sum_{v\in A_1}d(v,X)+\sum_{v\in A_2}\min(d(v,X),\alpha d(v,Y))+\alpha\sum_{v\in A_3\cup A_4} d(v,Y).
\end{align*}

The cost of the clustering induced by $X$ under $d'$, using our assumption, is equal to 
\begin{align*}
\sum_{v\in S} d'(v,X)\leq \sum_{v\in A_1} d(v,X)+\sum_{v\in A_2}\min(d(v,X),\alpha d(v,Y))+\alpha\sum_{v\in A_3\cup A_4} d(v,Y).
\end{align*}

Therefore, $X$ is the optimal set of centers under the perturbation $d'$.
Given an $\alpha$-LPR cluster $C_i$, by definition, there exists $x\in X$ such that $\text{Vor}_X(x)=C_i$ under $d'$,
therefore by construction, $\text{Vor}_X(x)=C_i$ under $d$ as well.
This proves the theorem.
\end{proof}

\begin{figure}
    \centering
    \begin{subfigure}[b]{0.5\textwidth}
        \includegraphics[width=\textwidth]{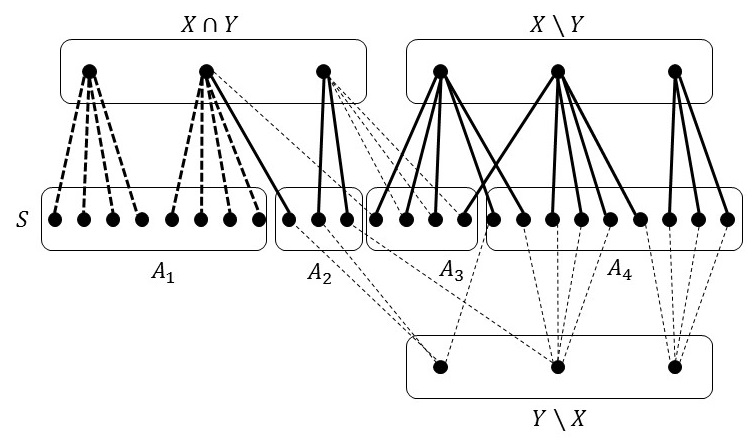}
        \caption{\small The bold lines represent center assignment in $X$, and 
the dotted lines represent center assignment in $Y$.}
        \label{fig:approx}
    \end{subfigure}
    \quad\quad\quad\quad
    \begin{subfigure}[b]{0.3\textwidth}
        \includegraphics[width=\textwidth]{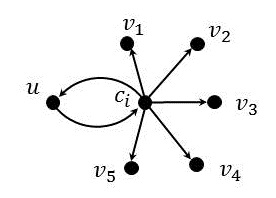}
        \caption{\small $u$ is a CCV, so it is distance $2r^*$ to its entire cluster.
				Each directed arrow represents a distance of $r^*$. 
				}
        \label{fig:ccv}
    \end{subfigure}
    \caption{The clustering setup for Lemma~\ref{lem:approx} and Theorem~\ref{thm:approx} (left),
		and an example of a center-capturing vertex (right).}
    \label{fig:approx-ccv}
\end{figure}

Essentially, Lemma~\ref{lem:approx} claims that any $\alpha$-approximation algorithm will return the $\alpha$-LPR clusters
as long as the approximation ratio is ``uniform'' across all clusters. For example, a 3-approximation algorithm that returns
half of the clusters paying 1.1 times their cost in $\OPT$, and half of the clusters paying 5 times their cost in $\OPT$, would
fail the property. Luckily, the local search algorithm is well-suited for this property, due to the local optimum guarantee.

The local search algorithm starts with any set of $k$ centers, and iteratively replaces $t$ centers with $t$ different centers if it
leads to a better clustering (see Algorithm~\ref{alg:local}).
The number of iterations is $O(\frac{n}{\epsilon})$.
The classic Local Search heuristic is widely used in practice, and many works have 
studied local search theoretically for $k$-median and $k$-means~\cite{arya2004local,gupta2008simpler,kanungo2002local}.
For more information on local search, see a general introduction by Aarts and Lenstra~\cite{aarts1997local}.

\begin{algorithm}[h]
\caption {Local Search algorithm for $k$-median clustering}
\label{alg:local}
\begin{algorithmic}[1]
\Require{$k$-median instance $(S,d)$, parameter $\epsilon$}
\State Pick an arbitrary set of centers $C$ of size $k$.
\State While $\exists C'$ of size $k$ such that $|C\setminus C'|+|C'\setminus C|\leq\frac{2}{\epsilon}$ and $\text{cost}(C')\leq(1-\frac{\epsilon}{n})\text{cost}(C)$
\begin{itemize}
\item Replace $C$ with $C'$.
\end{itemize}
\Ensure{Centers $C$}
\end{algorithmic}
\end{algorithm}

The next theorem utilizes Lemma~\ref{lem:approx} and a result by Cohen-Addad and Schwiegelshohn~\cite{cohen2017one}, who showed that local
search returns the optimal clustering under a stronger version of
$(3+2\epsilon)$-PR.
For the formal proof of 
Theorem~\ref{thm:approx}, see Appendix~\ref{app:approx}.

\begin{theorem}\label{thm:approx}
Given a $k$-median instance $(S,d)$, running local search with search size $\frac{1}{\epsilon}$ returns a clustering that contains
every $(3+2\epsilon)$-LPR cluster, and it gives a $(3+2\epsilon)$-approximation overall.
\end{theorem}

The value of local perturbation resilience, $3+2\epsilon$, is tight due to a counterexample by Cohen-Addad et al.~\cite{cohen2016local}.
For $k$-means, we can use local search to return each $(9+\epsilon)$-LPR cluster by using Lemma~\ref{lem:approx} (which works for squared distances as well)
and adding the result from Kanungo et al.~\cite{kanungo2002local}.

\paragraph{$k$-center}

Because the $k$-center objective takes the maximum rather than the average of all cluster costs, 
the equivalent of the condition in Lemma~\ref{lem:approx} is essentially satisfied by any $\alpha$-approximation algorithm.
We will now prove an even stronger result. Any $\alpha$-approximation for $k$-center returns the $\alpha$-LPR clusters, even
for metric perturbation resilience.
First we state a lemma which allows us to reason about a specific class of $\alpha$-perturbations which will be useful in this section as well as 
throughout the paper, for symmetric and asymmetric $k$-center.
For the full proofs, see Appendix~\ref{app:approx}.

\begin{lemma} \label{lem:d'_metric}
Given $\alpha\geq 1$ and an 
asymmetric $k$-center clustering instance $(S,d)$ with optimal radius $r^*$,
let $d''$ denote an $\alpha$-perturbation such that for all $u,v$, either
$d''(u,v)=\min(\alpha r^*,\alpha d(u,v))$ or $d''(u,v)=\alpha d(u,v)$.
Let $d'$ denote the metric completion of $d''$.
Then $d'$ is an $\alpha$-metric perturbation of $d$, and the optimal cost under $d'$ is $\alpha r^*$.
\end{lemma}

\begin{sproof}
First, $d'$ is a valid $\alpha$-metric perturbation of $d$ because for all $u,v$,
$d(u,v)\leq d'(u,v)\leq d''(u,v)\leq\alpha d(u,v)$.
To show the optimal cost under $d'$ is $\alpha r^*$, it suffices to prove that for all 
$u,v$, if $d(u,v)\geq r^*$, then $d'(u,v)\geq\alpha r^*$.
This is true of $d''$ by construction, so we show it still holds after taking the metric completion
of $d''$, which can shrink some distances.
Given $u,v$ such that $d(u,v)\geq r^*$, there exists a path 
$u=u_0$--$u_1$--$\cdots$--$u_{s-1}$--$u_s=v$
such that $d''(u,v)=\sum_{i=0}^{s-1}d'(u_i,u_{i+1})$ and
for all $0\leq i\leq s-1$, $d'(u_i,u_{i+1})=d''(u_i,u_{i+1})$.
If one of the segments has length $\geq r^*$ in $d$, then it has length $\geq\alpha r^*$ in $d''$ and we are done.
If not, all distances increase by exactly a factor of $\alpha$, so we sum up all distances to show $d'(u,v)\geq\alpha r^*$.
\end{sproof}

\begin{theorem} \label{thm:kcenter}
Given an asymmetric $k$-center clustering instance $(S,d)$ and an $\alpha$-approximate clustering $\mathcal{C}$,
each $\alpha$-LPR cluster is contained in $\mathcal{C}$, even under the weaker metric perturbation resilience condition.
\end{theorem}

\begin{sproof}
Similar to the proof of Lemma~\ref{lem:approx}, we construct an $\alpha$-perturbation $d'$ and argue that
$\mathcal{C}$ becomes the optimal clustering under $d'$. Let $r^*$ denote the optimal $k$-center radius of $(S,d)$.
First we define an $\alpha$-perturbation $d''$ by increasing the distance from each point $v\in S$ to its center $c$ in
$\mathcal{C}$ to $\min\{\alpha r^*,\alpha d(v,\mathcal{C}(v))\}$, and increase all other distances by a factor of $\alpha$.
Then by Lemma~\ref{lem:d'_metric}, the metric completion $d'$ of $d''$ has optimal cost $\alpha r^*$,
and so $\mathcal{C}$ is the optimal clustering. Now we finish off the proof in a manner identical to Lemma~\ref{lem:approx}.
\end{sproof}

We remark that Theorem~\ref{thm:kcenter} generalizes the result from Balcan et al.~\cite{symmetric}.
Although Theorem~\ref{thm:kcenter} applies more generally to asymmetric $k$-center, it is most useful for symmetric $k$-center,
for which there exist several 2-approximation algorithms~\cite{dyer1985simple,gonzalez1985clustering,hochbaum1985best}.
Asymmetric $k$-center is NP-hard to approximate to within a factor of $o(log^* n)$~\cite{chuzhoy2005asymmetric},
so Theorem~\ref{thm:kcenter} only guarantees returning the $O(\log^* n)$-LPR clusters. In the next section,
we show how to substantially improve this result.

\section{Asymmetric $k$-center} \label{sec:akc}

In this section, we show that a slight modification to the
$O(\log^* n)$ approximation algorithm of Vishwanathan~\cite{vishwanathan}
leads to an algorithm that maintains its performance in the worst case, while returning
each cluster $C_i$ with the following property: $C_i$ is 2-LPR, and all nearby clusters are 2-LPR as well.
This result also holds for metric perturbation resilience.
We start by formally giving the stronger version of Definition~\ref{def:lpr}.
Throughout this section, we denote the optimal $k$-center radius by $r^*$.

\begin{definition}
An optimal cluster $C_i$ satisfies \emph{$\alpha$-strong local perturbation resilience ($\alpha$-SLPR)} if 
for each $j$ such that there exists $u\in C_i$, $v\in C_j$ and $d(u,v)\leq r^*$,
then $C_j$ is $\alpha$-LPR.
\end{definition}

\begin{theorem} \label{thm:asy_local}
Given an asymmetric $k$-center clustering instance $(S,d)$ of size $n$, 
Algorithm \ref{alg:asy-pr} returns each 2-SLPR cluster exactly.
For each 2-LPR cluster $C_i$, Algorithm \ref{alg:asy-pr} outputs a cluster that is a superset of $C_i$
and does not contain any other 2-LPR cluster.
These statements hold for metric perturbation resilience as well.
Finally, the overall clustering returned by Algorithm \ref{alg:asy-pr} is an $O(\log^* n)$-approximation.
\end{theorem}

\paragraph{Approximation algorithm for asymmetric $k$-center}
We start with a recap of the $O(\log^* n)$-approximation algorithm of Vishwanathan~\cite{vishwanathan}.
This was the first nontrivial algorithm for asymmetric $k$-center, 
and the approximation ratio was later proven to be tight~\cite{chuzhoy2005asymmetric}.
To explain the algorithm, it is convenient to think of asymmetric $k$-center as a set covering problem.
Given an asymmetric $k$-center instance $(S,d)$,
define the directed graph $D_{(S,d)}=(S,A)$, where $A=\{(u,v)\mid d(u,v)\leq r^*\}$.
For a point $v\in S$, we define $\Gamma^+(v)$ and $\Gamma^-(v)$ as the set of vertices with an arc to and from $v$, respectively.
The asymmetric $k$-center problem is equivalent to finding a subset $C\subseteq S$ of size
$k$ such that $\cup_{c\in C}\Gamma^+(c)=S$.
We also define $\Gamma^-_x(v)$ and $\Gamma^+_x(v)$ as the set of vertices which have a path of length $\leq x$ 
to and from $v$ in $D_{(S,d)}$, respectively, and we define $\Gamma^+_x(A)=\bigcup_{v\in A}\Gamma^+_x(v)$ for a set $A\subseteq S$, 
and similarly for $\Gamma^-_x(A)$.
It is standard to assume the value of $r^*$ is known; since it is one of $O(n^2)$ distances, 
the algorithm can search for the correct value in polynomial time.
Vishwanathan's algorithm crucially utilizes the following concept.

\begin{definition}
Given an asymmetric $k$-center clustering instance $(S,d)$,
a point $v\in S$ is a \emph{center-capturing vertex} (CCV) if 
$\Gamma^-(v)\subseteq \Gamma^+(v)$.
In other words, for all $u\in S$, $d(u,v)\leq r^*$ implies $d(v,u)\leq r^*$.
\end{definition}

As the name suggests, each CCV $v\in C_i$, ``captures'' its center, i.e.\ $c_i\in\Gamma^+(v)$ (see Figure~\ref{fig:ccv}).
Therefore, $v$'s entire cluster is contained inside $\Gamma^+_2(v)$, which is a nice property that the
approximation algorithm exploits.
At a high level, the approximation algorithm has two phases. In the first phase, the algorithm iteratively picks
a CCV $v$ arbitrarily and removes all points in $\Gamma^+_2(v)$. 
This continues until there are no more CCVs.
For every CCV picked, the algorithm is guaranteed to remove an entire optimal cluster.
In the second phase, the algorithm runs $\log^* n$ rounds of
a greedy set-cover subroutine on the remaining points. 
See Algorithm \ref{alg:asy}.
To prove the second phase terminates in $O(\log^* n)$ rounds, the analysis crucially assumes there are no
CCVs among the remaining points. We refer the reader to~\cite{vishwanathan} for these details.

\begin{algorithm}[h]
\caption {\textsc{$O(\log^* n)$ approximation algorithm for asymmetric $k$-center~\cite{vishwanathan}}}
\label{alg:asy}
\begin{algorithmic}
\Require{Asymmetric $k$-center instance $(S, d)$, optimal radius $r^*$ (or try all possible candidates)}
\Statex Set $C=\emptyset$.
\Statex\textbf{Phase I: Pull out arbitrary CCVs}
\State While there exists an unmarked CCV
\begin{itemize}
\item Pick an unmarked CCV $c$, add $c$ to $C$, and mark all vertices in $\Gamma^+_{2}(c)$
\end{itemize}
\textbf{Phase II: Recursive set cover}
\State Set $A_0=S\setminus\Gamma^+_{5}(C)$, $i=0$. While $|A_i|>k$:
\begin{itemize}
\item Set $A'_{i+1}=\emptyset$. While there exists an unmarked point in $A_i$:
\begin{itemize}
\item Pick $v\in S$ which maximizes $\Gamma^+_{5}(v)\cap A_i$
\item Mark all points $\Gamma^+_{5}(v)\cap A_i$ and add $v$ to $A'_{i+1}$. 
\end{itemize}
\item Set $A_{i+1}=A'_{i+1}\cap A$ and $i=i+1$
\end{itemize}
\Ensure{Centers $C\cup A_{i+1}$}
\end{algorithmic}
\end{algorithm}

\paragraph{Robust algorithm for asymmetric $k$-center}
We show a small modification to Vishwanathan's approximation algorithm leads to simultaneous guarantees
in the worst case and under local perturbation resilience.
We show that each 2-LPR center is itself a CCV, and displays other structure which allows us to 
distinguish it from non-center CCVs. 
This suggests a simple modification to Algorithm~\ref{alg:asy}: 
instead of picking CCVs arbitrarily, we first pick CCVs which display the added structure,
and then when none are left, we go back to picking regular CCVs.
However, we need to ensure that a CCV chosen by the algorithm marks points from at most one 2-LPR cluster, 
or else we will not be able to output a separate cluster for each 2-LPR cluster.
Thus, the difficulty in our argument is carefully specifying which CCVs the algorithm picks, 
and which nearby points get marked by the CCVs, 
so that we do not mark other LPR clusters and \emph{simultaneously} maintain the 
guarantee of the original approximation algorithm,
namely that in every round, we mark an entire optimal cluster. 
To accomplish this tradeoff, we start by defining two properties.
The first property will determine which CCVs are picked by the algorithm.
The second property is used in the proof of correctness, but is not used
explicitly by the algorithm. We give the full details of the proofs in Appendix~\ref{app:akc}.

\begin{definition}
\emph{(1)} A point $c$ satisfies \emph{CCV-proximity} if it is a CCV, and each point in $\Gamma^-(c)$ 
is closer to $c$ than any CCV outside of $\Gamma^+(c)$. 
That is, for all points $v\in \Gamma^-(c)$ and CCVs $c'\notin \Gamma^+(c)$, 
$d(c,v)<d(c',v)$.
\footnote{
This property loosely resembles $\alpha$-center proximity~\cite{awasthi2012center},
a property defined over an entire clustering instance, 
which states for all $i$, for all $v\in C_i$, $j\neq i$, we have $\alpha d(c_i,v)<d(c_j,v)$.
}
\emph{(2)} An optimal center $c_i$ satisfies \emph{center-separation} if 
any point within distance $r^*$ of $c_i$ belongs to its cluster $C_i$. 
That is, for all $v\notin C_i$,  $c_i\notin\Gamma^+(v)$.
\end{definition}

\begin{lemma}\label{lem:properties}
Given an asymmetric $k$-center clustering instance $(S,d)$ and a 2-LPR cluster $C_i$,
$c_i$ satisfies CCV-proximity and center-separation.
Furthermore, given a CCV $c\in C_i$, a CCV $c'\notin C_i$, and a point $v\in C_i$, we have
$d(c,v)<d(c',v)$.
\end{lemma}

\begin{sproof}
Given an instance $(S,d)$ and a 2-LPR cluster $C_i$, we show that $c_i$ has the desired properties.

\emph{Center Separation:}
Assume there exists a point $v\in C_j$ for $j\neq i$ such that $d(v,c_i)\leq r^*$.
The idea is to construct a $2$-perturbation in which $v$ becomes the center for $C_i$,
since the distance from $v$ to each point in $C_i$ is $\leq 2r^*$ by the triangle inequality.
Define $d''$ by increasing all distances by a factor of 2, except for the distances between
$v$ and each point $u$ in $C_i$, which we increase to $\min(2r^*,2d(v,u))$.
By Lemma~\ref{lem:d'_metric}, the metric completion $d'$ of $d''$ is a 2-metric perturbation
with optimal cost $2r^*$, so we can replace $c_i$ with $v$ in the set of optimal centers under $d'$.
However, now $c_i$ switches to a different cluster, contradicting 2-LPR.

\emph{Final property:}
Given CCVs $c\in C_i$, $c'\in C_j$, and a point $v\in C_i$, assume $d(c',v)\leq d(c,v)$.
Again, we will use a perturbation to construct a contradiction. 
Since $c$ and $c'$ are CCVs and thus distance $2r^*$ to their clusters, we can construct a 2-metric perturbation
with optimal cost $2r^*$ in which $c$ and $c'$ become centers for their respective clusters.
Then $v$ switches clusters, so we have a contradiction.

\emph{CCV-proximity:}
By center-separation and the definition of $r^*$, we have that $\Gamma^-(c_i)\subseteq C_i\subseteq\Gamma^+(c_i)$,
so $c_i$ is a CCV.
Now given a point $v\in \Gamma^-(c_i)$ and a CCV $c\notin \Gamma^+(c_i)$,
from center-separation and definition of $r^*$,
$v\in C_i$ and $c\in C_j$ for $j\neq i$.
Then from the property in the previous paragraph, 
$d(c_i,v)<d(c,v)$.
\end{sproof}

Now we can modify the algorithm so that it first chooses CCVs satisfying CCV-proximity.
The other crucial change is instead of each chosen CCV $c$ marking all points in $\Gamma^+_2(c)$,
it instead marks all points $v$ such that $v\in\Gamma^+(c')$ for some $c'\in\Gamma^-(c)$.
See Algorithm~\ref{alg:asy-pr}.
Note this new way of marking preserves the guarantee that each CCV $c\ C_i$ marks its own cluster,
because $c_i\in\Gamma^-(c)$.
It also allows us to prove that each CCV $c$ satisfying CCV-proximity can never mark an LPR center $c_i$ from
a different cluster.
Intuitively, if $c$ marks $c_i$, then there exists a point $v\in\Gamma^-(c_i)\cap\Gamma^-(c)$,
but there can never exist a point $v$ distance $\leq r^*$ to two points satisfying CCV-proximity,
since both would need to be closer to $v$ by definition.
Finally, the last property in Lemma~\ref{lem:properties} allows us to prove that when the algorithm computes
the Voronoi tiles after Phase 1, all points will be correctly assigned.
Now we are ready to prove Theorem~\ref{thm:asy_local}.

\begin{algorithm}[h]
\caption {Robust algorithm for asymmetric $k$-center}
\label{alg:asy-pr}
\begin{algorithmic}
\Require{ Asymmetric $k$-center instance $(S, d)$, distance $r^*$ (or try all possible candidates)}
\State Set $C=\emptyset$. Redefine $d$ using the shortest path length in $D_{(S,d)}$,
breaking ties by distance to first common vertex in the shortest path.  \\
\textbf{Phase I: Pull out special CCVs}
\begin{itemize}
\item While there exists an unmarked CCV:
\begin{itemize}
\item Pick an unmarked point $c$ which satisfies CCV-proximity.
If no such $c$ exists, then pick an arbitrary unmarked CCV instead. Add $c$ to $C$.
\item For all points $c'\in \Gamma^-(c)$, mark all points in $\Gamma^+(c')$.
\end{itemize}
\item For each $c\in C$, let $V_c$ denote $c$'s Voronoi tile of the marked points
induced by $C$.
\end{itemize}
\textbf{Phase II: Recursive set cover}
\State Run Phase II as in Algorithm~\ref{alg:asy}, outputting $A_{i+1}$.
\State Compute the Voronoi tile for each center in $C\cup A_{i+1}$, 
but a point in $V_c$ must remain in $c$'s Voronoi tile.
\footnote{We note this step breaks the requirement that the outputted clustering
is a Voronoi partition. See the proof sketch of Theorem~\ref{thm:asy_local}}
\end{algorithmic}
\end{algorithm}

\begin{proof}[Proof sketch of Theorem~\ref{thm:asy_local}]
First we explain why Algorithm \ref{alg:asy-pr} retains the approximation guarantee of Algorithm \ref{alg:asy}.
Given any CCV $c\in C_i$ chosen in Phase I, $c$ marks its entire cluster by definition, and we start Phase II with no remaining CCVs. 
This condition is sufficient for Phase II to return an
$O(\log^* n)$ approximation (Theorem 3.1 from~\cite{vishwanathan}).

Next we claim that for each 2-LPR cluster $C_i$, there exists a cluster outputted by Algorithm~\ref{alg:asy-pr} 
that is a superset of $C_i$ and does not contain any other 2-LPR cluster.
To prove this claim, we first show there exists a point from $C_i$ satisfying CCV-proximity that cannot be marked by any point from
a different cluster in Phase I.
From Lemma~\ref{lem:properties}, $c_i$ satisfies CCV-proximity and
center-separation. If a point $c\notin C_i$ marks $c_i$, then
$\exists v\in\Gamma^-(c)\cap\Gamma^-(c_i)$.
By center-separation, $c\notin\Gamma^-(c_i)$. 
Then from the definition of CCV-proximity, both $c$ and $c_i$
must be closer to $v$ than the other, causing a contradiction.
At this point, we know a point $c\in C_i$ will always be chosen by the algorithm in Phase I.
To finish the claim, we show that each point $v$ from $C_i$ is closer to $c$ than to any other point $c'\notin C_i$ chosen in Phase I.
Since $c$ and $c'$ are both CCVs, this follows directly from Lemma~\ref{lem:properties}.
However, it is possible that a center $c'\in A_{i+1}$ is closer to $v$ than $c$ is to $v$, causing $c'$ to ``steal'' $v$; this is unavoidable.
Therefore, we forbid the algorithm from decreasing the size of the Voronoi tiles of $C$ after Phase I.

Finally, we claim that Algorithm \ref{alg:asy-pr} returns each 2-SLPR cluster exactly.
Given a 2-SLPR cluster $C_i$, by our previous argument, the algorithm chooses a CCV $c\in C_i$ such that $C_i\subseteq V_c$.
It is left to show that $V_c\subseteq C_i$.
The intuition is that since $C_i$ is 2-SLPR, its neighboring clusters were also marked in Phase I, and these clusters ``shield''
$V_c$ from picking up superfluous points in Phase II.
Specifically, there are two cases. If there exists $v\in V_c\subseteq C_i$ that was marked in Phase I, then we can prove that $v$
comes from a 2-LPR cluster, so $v\in V_c$ contradicts our previous argument.
If there exists $v\in V_c\subseteq C_i$ from Phase II, then the shortest path in $D_{(S,d)}$ from $c$ to $v$ is length at least 5 (see Figure~\ref{fig:slpr}).
The first point $u'\in C_j$, $j\neq i$ on the shortest path must come from a 2-LPR cluster, 
and we prove that $v$ is closer to $C_j$'s cluster using CCV-proximity.
\end{proof}

\begin{figure}
    \centering
    \begin{subfigure}[b]{0.45\textwidth}
        \includegraphics[width=\textwidth]{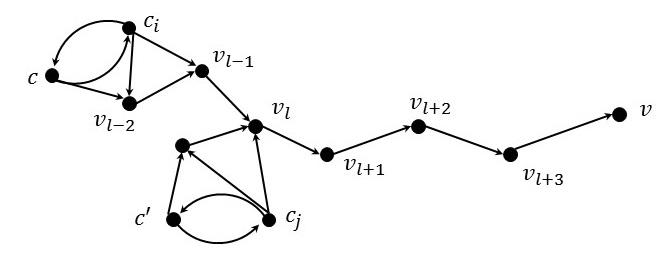}
        \caption{\small Proof of Theorem~\ref{thm:asy_local}. The arrows represent distances of $\leq r^*$.}
        \label{fig:slpr}
    \end{subfigure}
    \quad\quad\quad\quad
    \begin{subfigure}[b]{0.35\textwidth}
        \includegraphics[width=\textwidth]{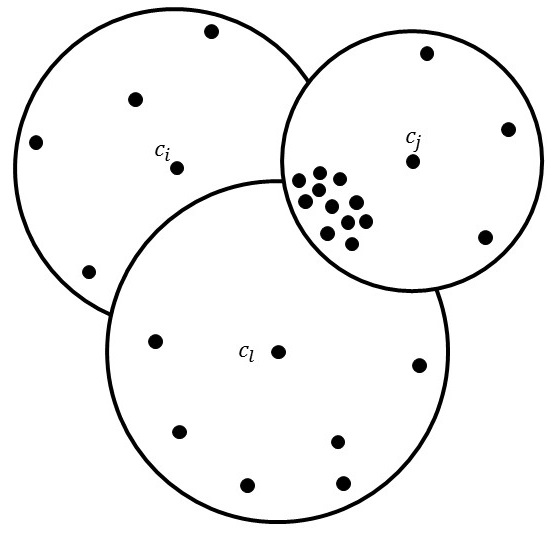}
        \caption{\small $c_i$ is a CCC for $C_j$, and $c_\ell$ is a CCC2 for $C_j$. The black disks have radius $r^*$.}
        \label{fig:ccc2}
    \end{subfigure}
    \caption{The proof of Theorem~\ref{thm:asy_local}, (left),
		and an example of a cluster-capturing center (right).}
    \label{fig:slpr-ccc2}
\end{figure}

\section{Robust local perturbation resilience} \label{sec:3eps}

In this section, we show that any 2-approximation algorithm for $k$-center returns
the optimal $(3,\epsilon)$-SLPR clusters, provided the
clusters are size $>2\epsilon n$.
Our main structural result is the following theorem.

\begin{theorem} \label{thm:3epsthm}
Given a $k$-center clustering instance $(S,d)$ with optimal radius $r^*$ such that all optimal clusters are size $>2\epsilon n$
and there are at least three $(3,\epsilon)$-LPR clusters,
then for each pair of $(3,\epsilon)$-LPR clusters $C_i$ and $C_j$,
for all $u\in C_i$ and $v\in C_j$, we have $d(u,v)>r^*$. 
\end{theorem}

Later in this section, we will show that both added conditions (the lower bound on the size of the clusters,
and that there are at least three $(3,\epsilon)$-LPR clusters) are necessary.
Since the distance from each point to its closest center is $\leq r^*$, a corollary 
of Theorem~\ref{thm:3epsthm} is that
any 2-approximate solution must contain the optimal $(3,\epsilon)$-SLPR clusters,
as long as the 2-approximation satisfies two sensible conditions:
\emph{(1)} for every edge $d(u,v)\leq 2r^*$ in
the 2-approximation, $\exists w$ s.t. $d(u,w)$ and $d(w,v)$ are $\leq r^*$, and \emph{(2)}
there cannot be multiple clusters outputted in the 2-approximation that can be combined into one cluster with the same radius.
Both of these properties are easily satisfied using quick pre- or post-processing steps.
\footnote{
For condition \emph{(1)}, before running the algorithm, remove all edges of distance $>r^*$, and then take the metric completion of
the resulting graph.
For condition \emph{(2)}, given the radius $\hat r$ of the
outputted solution, for each $v\in S$, check if the ball of radius $\hat r$ around $v$ captures multiple clusters. If so, combine them.}
We may also combine this result with Theorem \ref{thm:kcenter} to obtain a more powerful result for $k$-center.

\begin{theorem} \label{thm:kcenter_full}
Given a $k$-center clustering instance $(S,d)$ such that all optimal clusters are size $>2\epsilon n$
and there are at least three $(3,\epsilon)$-LPR clusters,
then any 2-approximate solution satisfying conditions \emph{(1)} and \emph{(2)} must contain all optimal 2-LPR clusters and $(3,\epsilon)$-SLPR
clusters.
\end{theorem}

\begin{proof}
Given such a clustering instance, then Theorem~\ref{thm:3epsthm} ensures that there is no edge of length $r^*$ between 
points from two different $(3,\epsilon)$-LPR clusters. Given a $(3,\epsilon)$-SLPR cluster $C_i$, it follows that there is no
point $v\notin C_i$ such that $d(v,C_i)\leq r^*$.
Therefore, given a 2-approximate solution $\mathcal{C}$ satisfying condition \emph{(1)}, any $u\in C_i$ and $v\notin C_i$ cannot be in the same
cluster. Furthermore, by condition \emph{(2)}, $C_i$ must not be split into two clusters.
Therefore, $C_i\in\mathcal{C}$.
The second part of the statement follows directly from Theorem \ref{thm:kcenter}.
\end{proof}

\paragraph{Proof idea for Theorem~\ref{thm:3epsthm}}
The proof consists of two parts. The first part is to show that if two points from
different LPR clusters are close together, then \emph{all} points in the clustering instance must be near each
other, in some sense (Lemma~\ref{lem:closepoints}).
The second part of the proof consists of showing that the points from three LPR clusters must be reasonably far from one
another; therefore, we achieve the final result by contradiction.

Here is the intution for part 1. Assume that there are points $u\in C_i$ and $v\in C_j$ 
from different LPR clusters, but $d(u,v)\leq r^*$.
Then by the triangle inequality, the distance from $u$ to $C_i\cup C_j$ is less than $3r^*$. We show that under a suitable
3-perturbation, we can replace $c_i$ and $c_j$ with $u$ in the set of optimal centers.
So, there is a 3-perturbation in which the optimal solution uses just $k-1$ centers.
However, as pointed out in~\cite{symmetric}, we are still a long way off from showing a contradiction.
Since the definiton of local perturbation resilience reasons about sets of $k$ centers, we must add a ``dummy center''.
But adding any point as a dummy center might not immediately result in a contradiction, if the voronoi partition
``accidentally'' outputs the LPR clusters.
To handle this problem, we use the notion of a \emph{cluster-capturing center}~\cite{symmetric},
intuitively, a center which is within $r^*$ of most of the points of a different optimal cluster (see Figure~\ref{fig:ccc2}).
This allows us to construct perturbations and control which points become centers for which clusters.
We show all of the points in the instance are close together, in some sense.

The second part of the argument diverges from all previous work in perturbation resilience,
since finding a contradiction under \emph{local} perturbation resilience poses a novel challenge.
From the previous part of the proof, we are able to find two noncenters $p$ and $q$, which are collectively close to all other points in the dataset.
Then we construct a 3-perturbation such that any size $k$ subset of $\{c_i\}_{i=1}^k\cup\{p,q\}$ is an optimal set of centers.
Our goal is to show that at least one of these subsets must break up a LPR cluster, causing a contradiction.
There are many cases to consider, so we build up conditional structural claims dictating the possible centers for each LPR cluster
under the 3-perturbation.
For instance, if a center $c_j$ is the best center for a LPR cluster $C_i$ under some set of optimal centers,
then $p$ or $q$ must be the best center for $C_j$, otherwise we would arrive at a contradiction by definition of LPR (Lemma~\ref{lem:rank-struct}).
We build up enough structural results to examine every possibility of center-cluster combinations, showing they all
lead to contradictions, thus negating our original assumption.

\paragraph{Formal analysis of Theorem~\ref{thm:3epsthm}}
Now we give the proof details for Theorem~\ref{thm:3epsthm}.
The first part of the proof resembles the argument for 
$(3,\epsilon)$-perturbation resilience by Balcan et al.~\cite{symmetric}.
We start with the following fact.

\begin{fact} \label{fact:half}
Given a $k$-center clustering instance $(S,d)$
such that all optimal clusters have size $>2\epsilon n$,
let $d'$ denote an $\alpha$-perturbation with optimal
centers $C'=\{c'_1,\dots,c'_k\}$.
Let $\mathcal{C'}$ denote the set of 
$(\alpha,\epsilon)$-LPR clusters.
Then there exists a one-to-one function $f:\mathcal{C'}\rightarrow C'$
such that for all $C_i\in\mathcal{C'}$, 
$f(C_i)$ is the center for more than half of the points in $C_i$ under $d'$.
\footnote{A non local version of this fact appeared in \cite{symmetric}.}
\end{fact}

In words, for any set of optimal centers under an $\alpha$-perturbation, 
each LPR cluster can be paired to a unique center.
This follows simply because all optimal clusters are size $>2\epsilon n$,
yet under a perturbation, $<\epsilon n$ points can switch out of each LPR cluster.
Next, we give the following definition, which will be a key point in the first part of the proof.

\begin{definition} \cite{symmetric} \label{def:CCC}
A center $c_i$ is a \emph{first-order cluster-capturing center} (CCC) 
for $C_j$ if for all $x\neq j$,
for all but $\epsilon n$ points
$v\in C_j$, $d(c_i,v)< d(c_x,v)$ and
$d(c_i,v)\leq r^*$.
$c_i$ is a \emph{second-order cluster-capturing center} (CCC2) 
for $C_j$ if there exists an $\ell$ such that for all $x\neq j,\ell$,
for all but $\epsilon n$ points
$v\in C_j$, $d(c_i,v)< d(c_x,v)$
and $d(c_i,v)\leq r^*$.
Then we say that $c_i$ is a CCC2 for $c_j$ discounting $c_\ell$.
See Figure~\ref{fig:ccc2}.
\end{definition}

Intuitively, a center $c_x$ is a CCC for $C_y$ if $c_x$ is a valid
center for $C_y$ when $c_y$ is removed from the set of optimal centers.
This is particularly useful when $C_y$ is $(\alpha,\epsilon)$-LPR,
since we can combine it with Fact~\ref{fact:half} to show that
$c_x$ is the unique center for the majority of points in $C_y$.
Another key idea in our analysis is the following concept.

\begin{definition} \label{def:hit}
A set $C\subseteq S$ $(\beta,\gamma)$-hits $S$ if for all $s\in S$,
there exist $\beta$ points in $C$ at distance $\leq\gamma r^*$ to $s$.
\end{definition}

We present the following lemma to demonstrate the usefulness of Definition~\ref{def:hit},
although this lemma will not be used until the second half of the proof of Theorem~\ref{thm:3epsthm}.

\begin{lemma} \label{lem:hit}
Given a $k$-center clustering instance $(S,d)$, given $z\geq 0$,
and given a set $C\subseteq S$ of size $k+z$ which $(z+1,\alpha)$-hits $S$,
there exists an $\alpha$-perturbation $d'$ such that all size $k$ subsets of $C$ are optimal sets of centers under $d'$.
\end{lemma}

\begin{proof}
Consider the following perturbation $d''$.

\begin{equation*}
d''(s,t)=
\begin{cases}
\min( \alpha r^*, \alpha d(s,t) ) & \text{if } s\in C \text{ and } d(s,t)\leq \alpha r^*  \\
\alpha d(s,t) & \text{otherwise.}
\end{cases}
\end{equation*}

This is an $\alpha$-perturbation by construction.
Define $d'$ as the metric completion of $d''$.
Then by Lemma \ref{lem:d'_metric}, $d'$ is an $\alpha$-metric perturbation with
optimal cost $\alpha r^*$. Given any size $k$ subset $C'\subseteq C$, 
then for all $v\in S$, there is still at least one $c\in C'$ such that
$d(c,v)\leq \alpha r^*$, therefore by construction, $d'(c,v)\leq \alpha r^*$.
It follows that $C'$ is a set of optimal centers under $d'$.
\end{proof}

Now we are ready to prove the first half of Theorem \ref{thm:3epsthm},
stated in the following lemma. For the full details, see Appendix~\ref{app:3eps}.

\begin{lemma} \label{lem:closepoints}
Given a $k$-center clustering instance $(S,d)$
such that all optimal clusters are size $>2\epsilon n$
and there exist two points at distance $r^*$
from different $(3,\epsilon)$-LPR clusters,
then there exists a partition $S_x\cup S_y$ of the non-centers $S\setminus\{c_\ell\}_{\ell=1}^k$ 
such that for all pairs $p\in S_x$, $q\in S_y$, $\{c_\ell\}_{\ell=1}^k\cup\{p,q\}$ $(3,3)$-hits $S$.
\end{lemma}

\begin{sproof}
This proof is split into two main cases. The first case is the following: there exists a CCC2 for a
$(3,\epsilon)$-LPR cluster, discounting a $(3,\epsilon)$-LPR cluster.
In fact, in this case, we do not need the assumption that two points from different LPR clusters are close.
If there exists a CCC to a $(3,\epsilon)$-LPR cluster,
denote the CCC by $c_x$ and the cluster by $C_y$. Otherwise, let $c_x$ denote a CCC2
to a $(3,\epsilon)$-LPR cluster $C_y$, discounting a $(3,\epsilon)$-LPR center $c_z$.
Then $c_x$ is at distance $\leq r^*$ to all but
$\epsilon n$ points in $C_y$. Therefore, $d(c_x,c_y)\leq 2r^*$ and so $c_x$ is at
distance $\leq 3r^*$ to all points in $C_y$.
Now we can create a 3-perturbation $d'$ by increasing all distances by a factor of 3 except for the distances
between $c_x$ and each point $v\in C_y$, which we increase to $\min(3r^*,3d(c_x,v))$.
Then by Lemma~\ref{lem:d'_metric}, $d'$ is a 3-perturbation with optimal cost $3 r^*$.
Therefore, given any non-center $v\in S$, the set of centers $\{c_\ell\}_{\ell=1}^k\setminus\{c_y\}\cup\{v\}$ achieves the optimal
score, and from Fact~\ref{fact:half}, one of the centers in 
$\{c_\ell\}_{\ell=1}^k\setminus\{c_y\}\cup\{v\}$ must be the center for the majority of points in $C_y$ under $d'$.
If this center is $c_\ell$, $\ell\neq x,y$, 
then by definition, $c_\ell$ is a CCC for the $(3,\epsilon)$-LPR cluster, $C_y$, which creates a contradiction because $\ell\neq x$.
Therefore, either $v$ or $c_x$ must be the center for the majority of points in $C_y$ under $d'$.

If $c_x$ is the center for the majority of points in $C_y$, 
then because $C_y$ is $(3,\epsilon)$-LPR, the corresponding cluster must contain
fewer than $\epsilon n$ points from $C_x$. 
Furthermore, since for all $\ell\neq x$ and $u\in C_x$, $d(u,c_x)<d(u,c_\ell)$,
it follows that $v$ must be the center
for the majority of points in $C_x$.
Therefore, every non-center $v\in S$ is at distance $\leq r^*$ to the majority of points in
either $C_x$ or $C_y$.

Now partition all the non-centers into two sets $S_x$ and $S_y$, such that
$S_x=\{p\mid \text{for the majority of points }q\in C_x,\text{ }d(p,q)\leq r^*\}$ and
$S_y=\{p\mid p\notin S_x\text{ and for the majority of points }q\in C_y,\text{ }d(p,q)\leq r^*\}$.
Given $p,q\in S_x$, then $d(p,q)\leq 2r^*$ since both points are close to more than half the points
in $C_x$. Similarly, any two points $p,q\in S_y$ are $\leq 2r^*$ apart.

Now we claim that $\{c_\ell\}_{\ell=1}^k\cup\{p,q\}$ $(3,3)$-hits $S$ for any pair $p\in S_x$, $q\in S_y$.
This is because a point $v\in C_i$ from $S_x$ is $3r^*$ to $p$, $c_i$, and $c_x$
and a point $v\in C_x$ is $3r^*$ to $c_x$, $c_y$, and $p$.
The latter follows because $d(c_x,c_y)\leq 2r^*$.
Similar statements are true for $S_y$ and $C_y$.

Now we turn to the other case.
Assume there does not exist a CCC2 to a LPR cluster, discounting a LPR center. 
In this case, we need to use the assumption that there exist $(3,\epsilon)$-LPR clusters $C_x$ and $C_y$,
and $p\in C_x$, $q\in C_y$ such that $d(p,q)\leq r^*$.
Then by the triangle inequality, $p$ is distance $\leq 3r^*$
to all points in $C_x$ and $C_y$.
Again we construct a 3-perturbation $d'$ by increasing all distances by a factor of 3 except for the distances
between $p$ and $v\in C_x\cup C_y$, which we increase to $\min(3r^*,3d(s,t))$.
By Lemma \ref{lem:d'_metric}, $d'$ has optimal cost $3 r^*$.
Then given any non-center $s\in S$, the set of centers $\{c_\ell\}_{\ell=1}^k\setminus\{c_x,c_y\}\cup\{p,s\}$ achieves the optimal
score.

From Fact~\ref{fact:half}, one of the centers in 
$\{c_\ell\}_{\ell=1}^k\setminus\{c_x,c_y\}\cup\{p,s\}$ must be the center for
the majority of points in $C_x$ under $d'$.
If this center is $c_\ell$ for $\ell\neq x,y$, 
then $c_\ell$ is a CCC2 for $C_x$ discounting $c_y$, which contradicts our assumption.
Similar logic applies to the center for the majority of points in $C_y$.
Therefore, $p$ and $s$ must be the centers for $C_x$ and $C_y$.
Since $s$ was an arbitrary non-center, all non-centers are 
distance $\leq r^*$ to all but $\epsilon n$ points in either $C_x$ or $C_y$.

Similar to Case 1, we now partition all the non-centers into two sets $S_x$ and $S_y$, such that
$S_x=\{u\mid \text{for the majority of points }v\in C_x,\text{ }d(u,v)\leq r^*\}$ and
$S_y=\{u\mid u\notin S_x\text{ and for the majority of points }v\in C_y,\text{ }d(u,v)\leq r^*\}$.
As before, each pair of points in $S_x$ are distance $\leq 2r^*$ apart, and similarly for $S_y$.

Again we must show that $\{c_\ell\}_{\ell=1}^k\cup\{p,q\}$ $(3,3)$-hits $S$ for each pair $p\in S_x$, $q\in S_y$.
It is no longer true that $d(c_x,c_y)\leq 2r^*$, however, we can prove that for both $S_x$ and $S_y$,
there exist points from two distinct clusters each. From the previous paragraph, given a non-center
$s\in C_i$ for $i\neq x,y$, we know that $p$ and $s$ are centers for $C_x$ and $C_y$.
With an identical argument, given $t\in C_j$ for $j\neq x,y,i$, we can show that $q$ and $t$ are centers for $C_x$ and $C_y$.
It follows that $S_x$ and $S_y$ both contain points from at least two distinct clusters.
Now given a non-center $s\in C_i$, WLOG $s\in S_x$, then there exists $j\neq i$ and $t\in C_j\cap S_x$.
Then $c_i$, $c_j$, and $p$ are $3r^*$ to $s$ and $c_i$, $c_x$, and $p$ are $3r^*$ to $c_i$.
In the case where $i=x$, then $c_i$, $c_j$, and $p$ are $3r^*$ to $c_i$.
This concludes the proof.
\end{sproof}

Now we move to the second half of the proof of Theorem \ref{thm:3epsthm}.
Recall that our goal is to show a contradiction assuming two points from different LPR clusters are close.
From Lemma~\ref{lem:hit} and Lemma~\ref{lem:closepoints},
we know there is a set of $k+2$ points, and any size $k$ subset is optimal under a suitable perturbation.
And by Lemma~\ref{fact:half}, each size $k$ subset must have a mapping from LPR clusters to centers.
Now we state a fact which states these mappings are derived from a ranking of all possible center points by the LPR clusters.
In other words, each LPR cluster $C_x$ can rank all the points in $S$, so that for any set of optimal
centers for an $\alpha$-perturbation, the top-ranked center is the one whose cluster is $\epsilon$-close to $C_x$.
We defer the proof to Appendix~\ref{app:3eps}.

\begin{fact} \label{fact:ranking}
Given a $k$-center clustering instance $(S,d)$
such that all optimal clusters have size $>2\epsilon n$,
and an $\alpha$-perturbation $d'$ of $d$,
let $\mathcal{C'}$ denote the set of 
$(\alpha,\epsilon)$-LPR clusters.
For each $C_x\in\mathcal{C'}$, there exists a bijection $R_{x,d'}:S\rightarrow[n]$
such that for all sets of $k$ centers $C$ that achieve the optimal cost under $d'$, 
then $c=\text{argmin}_{c'\in C} R_{x,d'}(c')$ if and only if $\text{Vor}_C(c)$ is $\epsilon$-close to $C_x$.
\end{fact}

For an LPR cluster $C_x$ and a subset $S'\subseteq S$ of size $n'$, 
we also define $R_{x,d',S'}:S'\rightarrow [n']$
as the ranking specific to $S'$.
Now, using Fact~\ref{fact:ranking} with the previous Lemmas, we can try to give a contradiction by showing that there is no
set of rankings for the LPR clusters that is consistent with all the optimal sets of centers guaranteed by Lemmas~\ref{lem:hit} and~\ref{lem:closepoints}.
The following lemma gives relationships among the possible rankings.
These will be our main tools for contradicting LPR and thus finishing the proof of Theorem~\ref{thm:3epsthm}.
Again, the proof is provided in Appendix~\ref{app:3eps}.

\begin{lemma} \label{lem:rank-struct}
Given a $k$-center clustering instance $(S,d)$
such that all optimal clusters are size $>2\epsilon n$, 
and given non-centers $p,q\in S$ such that
$C=\{c_\ell\}_{\ell=1}^k\cup\{p,q\}$ $(3,3)$-hits $S$,
let the set $\mathcal{C'}$ denote the set of 
$(3,\epsilon)$-LPR clusters.
Define the 3-perturbation $d'$ as in Lemma \ref{lem:hit}. 
The following are true.
\begin{enumerate}
\item Given $C_x\in \mathcal{C'}$ and $C_i$ such that $i\neq x$, $R_{x,d'}(c_x)<R_{x,d'}(c_i)$.
\item There do not exist $s\in C$ and $C_x,C_y\in\mathcal{C'}$ such that $x\neq y$, and $R_{x,d',C}(s)+R_{y,d',C}(s)\leq 4$.
\item Given $C_i$ and $C_x,C_y\in\mathcal{C'}$ such that $x\neq y\neq i$,
if $R_{x,d',C}(c_i)\leq 3$, then $R_{y,d',C}(p)\geq 3$ and $R_{y,d',C}(q)\geq 3$.
\end{enumerate}
\end{lemma}

We are almost ready to bring everything together to give a contradiction. Recall that Lemma~\ref{lem:closepoints}
allows us to choose a pair $(p,q)$ such that $\{c_\ell\}_{\ell=1}^k\cup\{p,q\}$ $(3,3)$-hits $S$.
For an arbitrary choice of $p$ and $q$, we may not end up with a contradiction.
It turns out, we will need to make sure one of the points comes from an LPR cluster, and is very high in the ranking
list of its own cluster. This motivates the following fact, which is the final piece to the puzzle.

\begin{fact} \label{fact:prox}
Given a $k$-center clustering instance $(S,d)$
such that all optimal clusters are size $>2\epsilon n$, 
given an $(\alpha,\epsilon)$-LPR cluster $C_x$, and given $i\neq x$,
then there are fewer than $\epsilon n$ points $s\in C_x$ such that $d(c_i,s)\leq\min(r^*,\alpha d(c_x,s))$.
\end{fact}

\begin{proof}
Assume the fact is false. Then let $B\subseteq C_x$ denote a set of size $\epsilon n$ such that for all $s\in B$, 
$d(c_i,s)\leq\min(r^*,\alpha d(c_x,s))$.
Construct the following perturbation $d'$. For all $s\in B$, set $d'(c_x,s)=\alpha d(c_x,s)$. 
For all other pairs $s,t$, set $d'(s,t)=d(s,t)$. This is clearly an $\alpha$-perturbation by construction.
Then the original set of optimal centers still achieves cost $r^*$ under $d'$ because for all $s\in B$, $d'(c_i,s)\leq r^*$.
Clearly, the optimal cost under $d'$ cannot be $<r^*$. It follows that the original set of optimal centers $C$ is still optimal under $d'$.
However, all points in $B$ are no longer in $\text{Vor}_C(c_x)$ under $d'$, contradicting the fact that $C_x$ is $(\alpha,\epsilon)$-LPR.
\end{proof}

Now we are ready to prove Theorem~\ref{thm:3epsthm}.

\begin{proof} [Proof of Theorem~\ref{thm:3epsthm}]
Assume towards contradiction that there are two points at distance $\leq r^*$ from different
$(3,\epsilon)$-LPR clusters.
Then by Lemma~\ref{lem:closepoints},
there exists a partition $S_1,S_2$ of non-centers of $S$ such that for all 
pairs $p\in S_1$, $q\in S_2$, $\{c_\ell\}_{\ell=1}^k\cup\{p,q\}$ $(3,3)$-hit $S$.
Given three $(3,\epsilon)$-LPR clusters $C_x$, $C_y$, and $C_z$,
let $c_x'$, $c_y'$, and $c_z'$ denote the optimal centers ranked highest by $C_x$, $C_y$, and $C_z$ disregarding
$c_x$, $c_y$, and $c_z$, respectively.
Define $p=\text{argmin}_{s\in C_x} d(c_i,s)$, and WLOG let $p\in S_1$. Then pick an arbitrary point $q$ from $S_2$, and define $C=\{c_\ell\}_{\ell=1}^k\cup\{p,q\}$.
Define $d'$ as in Lemma \ref{lem:hit}. 
We claim that $R_{x,d',C}(p)<R_{x,d',C}(c_x')$:
from Fact~\ref{fact:prox}, there are fewer than $\epsilon n$ points $s\in C_x$ such that $d(c_x',s)\leq\min(r^*,3 d(c_x,s))$.
Among each remaining point $s\in C_x$, we will show $d'(p,s)\leq d'(c_x',s)$. Recall that $d(p,s)\leq d(p,c_x)+d(c_x,s)\leq 2r^*$, so $d'(p,s)=\min(3r^*,3d(p,s))$. 
There are two cases to consider.

Case 1: $d(c_x',s)>r^*$. Then by construction, $d'(c_x',s)\geq 3r^*$, and so $d'(p,s)\leq d'(c_x',s)$.

Case 2: $3d(c_x,s)<d(c_x',s)$. Then $d'(p,s)\leq 3d(p,s)\leq 3(d(p,c_x)+d(c_x,s))\leq 6d(c_x,s)\leq 2d(c_x',s)\leq \min(3r^*,3d(c_x',s))=d'(c_x',s)$,
and this proves our claim.

Because $R_{x,d',C}(p)<R_{x,d',C}(c_x')$, it follows that either $R_{x,d',C}(p)\leq 2$ or $R_{x,d',C}(q)\leq 2$,
since the top two can only be $c_x$, $p$, or $q$.
The rest of the argument is broken up into cases.

Case 1: $R_{x,d',C}(c_x')\leq 3$. 
From Lemma~\ref{lem:rank-struct}, then $R_{y,d',C}(p)\geq 3$ and $R_{y,d',C}(q)\geq 3$.
It follows by process of elimination that 
$R_{y,d',C}(c_{y})=1$ and $R_{y,d',C}(c_{y'})=2$.
Again by Lemma~\ref{lem:rank-struct}, $R_{x,d',C}(p)\geq 3$ and $R_{x,d',C}(q)\geq 3$, causing a contradiction. 

Case 2: $R_{x,d',C}(c_{x'})>3$ and $R_{y,d',C}(c_{y'})\leq 3$.
Then $R_{x,d',C}(p)\leq 3$ and $R_{x,d',C}(q)\leq 3$. 
From Lemma~\ref{lem:rank-struct}, $R_{x,d',C}(p)\geq 3$ and $R_{x,d',C}(q)\geq 3$, therefore we have a contradiction.
Note, the case where $R_{x,d',C}(c_{x'})>3$ and $R_{z,d',C}(c_{z'})\leq 3$ is identical to this case.

Case 3: The final case is when $R_{x,d',C}(c_{x'})>3$, $R_{y,d',C}(c_{y'})>3$, and $R_{z,d',C}(c_{z'})>3$.
So for each $i\in\{x,y,z\}$, the top three for $C_i$ in $C$ is a permutation of $\{c_i,p,q\}$.
Then each $i\in\{x,y,z\}$ must rank $p$ or $q$ in the top two, so by the Pigeonhole Principle, 
either $p$ or $q$ is ranked top two by
two different LPR clusters, contradicting Lemma~\ref{lem:rank-struct}.
This completes the proof.
\end{proof}

We note that Case 3 in Theorem~\ref{thm:3epsthm} is the reason why we need to assume there are at least three
$(3,\epsilon)$-LPR clusters. If there are only two, $C_x$ and $C_y$, it is possible that there exist
$u\in C_x$, $v\in C_y$ such that $d(u,v)\leq r^*$. 
In this case, for $p,q,d',$ and $C$ as defined in the proof of
Theorem~\ref{thm:3epsthm}, if $C_x$ ranks $c_x$, $p$, $q$ as its top three and $C_y$ ranks $c_y$, $q$, $p$
as its top three, then there is no contradiction.

\paragraph{APX-Hardness under approximation stability}
Now we show the lower bound on the cluster sizes in Theorem~\ref{thm:kcenter_full} is necessary,
by showing hardness of approximation even when it is guaranteed the clustering satisfies 
$(\alpha,\epsilon)$-perturbation resilience for $\alpha\geq 1$ and $\epsilon>0$. 
In fact, this hardness holds even under the strictly stronger notion of 
\emph{approximation stability}~\cite{as}. 
We say that two clusterings $\mathcal{C}$ and $\mathcal{C}'$ are $\epsilon$-close if only
an $\epsilon$-fraction of the input points are clustered differently, i.e.,
$\min_\sigma\sum_{i=1}^k |C_i\setminus C_{\sigma(i)}'|\leq\epsilon n$.

\begin{definition}
A clustering instance $(S,d)$
satisfies \emph{$(\alpha,\epsilon)$-approximation stability ($(\alpha,\epsilon)$-AS)}
if any clustering $\mathcal{C}'$ (not necessarily a Voronoi partition)
such that $\text{cost}(\mathcal{C}')\leq\alpha\text{cost}(\mathcal{OPT})$
is $\epsilon$-close to $\mathcal{OPT}$.
\end{definition}

The hardness is based on a reduction from the general clustering
instances, so the APX-hardness constants match the non-stable APX-hardness results.

\begin{theorem} \label{thm:lowerbound}
Given  $\alpha\geq 1$, $\epsilon>0$,
it is NP-hard to approximate $k$-center to 2, $k$-median to $1.73$, or $k$-means to 1.0013,
even when it is guaranteed the instance satisfies $(\alpha,\epsilon)$-approximation stability.
\end{theorem}

This theorem generalizes hardness results from~\cite{as} and~\cite{symmetric}.
Also, because of Fact~\ref{fact:local-iff}, a corollary is that
unless $P=NP$, there is no efficient algorithm which outputs each $(\alpha,\epsilon)$-LPR cluster
for $k$-center (showing the condition on the cluster sizes in Theorem~\ref{thm:kcenter_full}
is necessary).

\begin{proof}
Given $\alpha\geq 1$, $\epsilon>0$, assume there exists a $\beta$-approximation algorithm $\mathcal{A}$ for $k$-median under
$(\alpha,\epsilon)$-approximation stability. We will show a reduction to $k$-median without approximation stability.
Given a $k$-median clustering instance $(S,d)$ of size $n$, we will create a new instance $(S',d')$ for $k'=k+n/\epsilon$ with
size $n'=n/\epsilon$ as follows.
First, set $S'=S$ and $d'=d$, and then add $n/\epsilon$ new points to $S'$, such that their distance to every other point is
$2\alpha n\max_{u,v\in S}d(u,v)$. Let $\OPT$ denote the optimal solution of $(S,d)$. Then the optimal solution to $(S',d')$ is
to use $\OPT$ for the vertices in $S$, and make each of the $n/\epsilon$ added points a center. 
Note that the cost of $\OPT$ and the optimal clustering for $(S',d')$ are identical, since the added points are distance 0 to their center.
Given a clustering $\mathcal{C}$ on $(S,d)$, let $\mathcal{C}'$ denote the clustering of $(S',d')$ that clusters $S$ as in $\mathcal{C}$, and then adds
$n/\epsilon$ extra centers on each of the added points. Then the cost of $\mathcal{C}$ and $\mathcal{C}'$ are the same, so it follows that
$\mathcal{C}$ is a $\beta$-approximation to $(S,d)$ if and only if $\mathcal{C}'$ is a $\beta$-approximation to $(S',d')$.
Next, we claim that $(S',d')$ satisfies $(\alpha,\epsilon)$-approximation stability.
Given a clustering $\mathcal{C}'$ which is an $\alpha$-approximation to $(S',d')$, then there must be a center located at all
$n/\epsilon$ of the added points, otherwise the cost of $\mathcal{C}'$ would be $>\alpha\OPT$.
Therefore, $\mathcal{C}'$ agrees with the optimal solution on all points except for $S$, therefore, $\mathcal{C}'$ must be $\epsilon$-close to
the optimal solution.
Now that we have established a reduction, the theorem follows from hardness of $1.73$-approximation for $k$-median~\cite{jain2002new}.
The proofs for $k$-center and $k$-means are identical, using hardness from~\cite{gonzalez1985clustering} and~\cite{lee2017improved}, respectively.
\end{proof}

\section{Conclusion} \label{sec:conclusion}

In this work, we initiate the study of clustering under local stability.
We define local perturbation resilience, a property of a single
optimal cluster rather than the instance as a whole.
We give algorithms that simultaneously achieve guarantees in the worst case, as well as
guarantees when the data is stable.

Specifically, we show that local search outputs the optimal $(3+\epsilon)$-LPR clusters 
for $k$-median and the $(9+\epsilon)$-LPR clusters for $k$-means. 
For $k$-center, we show that any 2-approximation outputs the optimal 2-LPR clusters, as well as the optimal $(3,\epsilon)$-LPR
clusters when assuming the optimal clusters are not too small.
We provide a natural modification to the asymmetric $k$-center approximation algorithm of Vishwanathan~\cite{vishwanathan}
to prove it outputs all 2-SLPR clusters.
Finally, we show APX-hardness of clustering under $(\alpha,\epsilon)$-approximation stability
for any $\alpha\geq 1$, $\epsilon>0$.
It would be interesting to find other approximation algorithms satisfying the condition in Lemma~\ref{lem:approx},
and in general to further study local stability for other objectives and conditions.

\clearpage


\bibliography{clustering}

\begin{thebibliography}{10}

\bibitem{aarts1997local}
E~Aarts and JK~Lenstra.
\newblock Local search in combinatorial optimization.
\newblock {\em John Wiley \& Sons, Inc.}, 1997.

\bibitem{ahmadian2016better}
Sara Ahmadian, Ashkan Norouzi-Fard, Ola Svensson, and Justin Ward.
\newblock Better guarantees for k-means and euclidean k-median by primal-dual
  algorithms.
\newblock {\em CoRR}, abs/1612.07925, 2016.

\bibitem{angelidakis2017algorithms}
Haris Angelidakis, Konstantin Makarychev, and Yury Makarychev.
\newblock Algorithms for stable and perturbation–resilient problems.
\newblock In {\em Proceedings of the Annual Symposium on Theory of Computing
  (STOC)}, 2017.

\bibitem{arora2012}
Sanjeev Arora, Rong Ge, and Ankur Moitra.
\newblock Learning topic models - going beyond {SVD}.
\newblock In {\em Proceedings of the Annual Symposium on Foundations of
  Computer Science (FOCS)}, pages 1--10, 2012.

\bibitem{arya2004local}
Vijay Arya, Naveen Garg, Rohit Khandekar, Adam Meyerson, Kamesh Munagala, and
  Vinayaka Pandit.
\newblock Local search heuristics for k-median and facility location problems.
\newblock {\em SIAM Journal on Computing}, 33(3):544--562, 2004.

\bibitem{awasthi2010stability}
Pranjal Awasthi, Avrim Blum, and Or~Sheffet.
\newblock Stability yields a ptas for k-median and k-means clustering.
\newblock In {\em Proceedings of the Annual Symposium on Foundations of
  Computer Science (FOCS)}, pages 309--318, 2010.

\bibitem{awasthi2012center}
Pranjal Awasthi, Avrim Blum, and Or~Sheffet.
\newblock Center-based clustering under perturbation stability.
\newblock {\em Information Processing Letters}, 112(1):49--54, 2012.

\bibitem{awasthi2015hardness}
Pranjal Awasthi, Moses Charikar, Ravishankar Krishnaswamy, and Ali~Kemal Sinop.
\newblock The hardness of approximation of euclidean k-means.
\newblock {\em Proceedings of the Annual Symposium on Computational Geometry
  (SOCG)}, 2015.

\bibitem{Awasthi2012Improved}
Pranjal Awasthi and Or~Sheffet.
\newblock Improved spectral-norm bounds for clustering.
\newblock In {\em Proceedings of the International Workshop on Approximation,
  Randomization, and Combinatorial Optimization Algorithms and Techniques
  (APPROX-RANDOM)}, pages 37--49. Springer, 2012.

\bibitem{as}
Maria-Florina Balcan, Avrim Blum, and Anupam Gupta.
\newblock Clustering under approximation stability.
\newblock {\em Journal of the ACM (JACM)}, 60(2):8, 2013.

\bibitem{symmetric}
Maria-Florina Balcan, Nika Haghtalab, and Colin White.
\newblock $ k $-center clustering under perturbation resilience.
\newblock In {\em Proceedings of the Annual International Colloquium on
  Automata, Languages, and Programming (ICALP)}, 2016.

\bibitem{balcan2012clustering}
Maria~Florina Balcan and Yingyu Liang.
\newblock Clustering under perturbation resilience.
\newblock {\em SIAM Journal on Computing}, 45(1):102--155, 2016.

\bibitem{balcan2009agnostic}
Maria~Florina Balcan, Heiko R{\"o}glin, and Shang-Hua Teng.
\newblock Agnostic clustering.
\newblock In {\em International Conference on Algorithmic Learning Theory},
  pages 384--398, 2009.

\bibitem{bilu2012stable}
Yonatan Bilu and Nathan Linial.
\newblock Are stable instances easy?
\newblock {\em Combinatorics, Probability and Computing}, 21(05):643--660,
  2012.

\bibitem{byrka2015improved}
Jaros{\l}aw Byrka, Thomas Pensyl, Bartosz Rybicki, Aravind Srinivasan, and Khoa
  Trinh.
\newblock An improved approximation for k-median, and positive correlation in
  budgeted optimization.
\newblock In {\em Proceedings of the Annual Symposium on Discrete Algorithms
  (SODA)}, pages 737--756, 2015.

\bibitem{charikar1999constant}
Moses Charikar, Sudipto Guha, {\'E}va Tardos, and David~B Shmoys.
\newblock A constant-factor approximation algorithm for the k-median problem.
\newblock In {\em Proceedings of the Annual Symposium on Theory of Computing
  (STOC)}, pages 1--10, 1999.

\bibitem{outliers}
Moses Charikar, Samir Khuller, David~M Mount, and Giri Narasimhan.
\newblock Algorithms for facility location problems with outliers.
\newblock In {\em Proceedings of the Annual Symposium on Discrete Algorithms
  (SODA)}, pages 642--651, 2001.

\bibitem{chen2008constant}
Ke~Chen.
\newblock A constant factor approximation algorithm for k-median clustering
  with outliers.
\newblock In {\em Proceedings of the Annual Symposium on Discrete Algorithms
  (SODA)}, pages 826--835, 2008.

\bibitem{chuzhoy2005asymmetric}
Julia Chuzhoy, Sudipto Guha, Eran Halperin, Sanjeev Khanna, Guy Kortsarz,
  Robert Krauthgamer, and Joseph~Seffi Naor.
\newblock Asymmetric k-center is log* n-hard to approximate.
\newblock {\em Journal of the ACM (JACM)}, 52(4):538--551, 2005.

\bibitem{cohen2016local}
Vincent Cohen-Addad, Philip~N Klein, and Claire Mathieu.
\newblock Local search yields approximation schemes for k-means and k-median in
  euclidean and minor-free metrics.
\newblock In {\em Proceedings of the Annual Symposium on Foundations of
  Computer Science (FOCS)}, pages 353--364, 2016.

\bibitem{cohen2017one}
Vincent Cohen{-}Addad and Chris Schwiegelshohn.
\newblock On the local structure of stable clustering instances.
\newblock {\em CoRR}, abs/1701.08423, 2017.

\bibitem{dyer1985simple}
Martin~E Dyer and Alan~M Frieze.
\newblock A simple heuristic for the p-centre problem.
\newblock {\em Operations Research Letters}, 3(6):285--288, 1985.

\bibitem{gonzalez1985clustering}
Teofilo~F Gonzalez.
\newblock Clustering to minimize the maximum intercluster distance.
\newblock {\em Theoretical Computer Science}, 38:293--306, 1985.

\bibitem{gupta2008simpler}
Anupam Gupta and Kanat Tangwongsan.
\newblock Simpler analyses of local search algorithms for facility location.
\newblock {\em CoRR}, abs/0809.2554, 2008.

\bibitem{gupta2014decompositions}
Rishi Gupta, Tim Roughgarden, and C~Seshadhri.
\newblock Decompositions of triangle-dense graphs.
\newblock In {\em Proceedings of the Annual Conference on Innovations in
  Theoretical Computer Science (ITCS)}, pages 471--482, 2014.

\bibitem{hardt2013beyond}
Moritz Hardt and Aaron Roth.
\newblock Beyond worst-case analysis in private singular vector computation.
\newblock In {\em Proceedings of the Annual Symposium on Theory of Computing
  (STOC)}, pages 331--340, 2013.

\bibitem{hochbaum1985best}
Dorit~S Hochbaum and David~B Shmoys.
\newblock A best possible heuristic for the k-center problem.
\newblock {\em Mathematics of operations research}, 10(2):180--184, 1985.

\bibitem{jain2002new}
Kamal Jain, Mohammad Mahdian, and Amin Saberi.
\newblock A new greedy approach for facility location problems.
\newblock In {\em Proceedings of the Annual Symposium on Theory of Computing
  (STOC)}, pages 731--740, 2002.

\bibitem{kanungo2002local}
Tapas Kanungo, David~M Mount, Nathan~S Netanyahu, Christine~D Piatko, Ruth
  Silverman, and Angela~Y Wu.
\newblock A local search approximation algorithm for k-means clustering.
\newblock In {\em Proceedings of the Annual Symposium on Computational Geometry
  (SOCG)}, pages 10--18. ACM, 2002.

\bibitem{kumar2010clustering}
Amit Kumar and Ravindran Kannan.
\newblock Clustering with spectral norm and the k-means algorithm.
\newblock In {\em Proceedings of the Annual Symposium on Foundations of
  Computer Science (FOCS)}, pages 299--308, 2010.

\bibitem{kumar2004simple}
Amit Kumar, Yogish Sabharwal, and Sandeep Sen.
\newblock A simple linear time (1+ $\varepsilon$)-approximation algorithm for
  geometric k-means clustering in any dimensions.
\newblock In {\em Proceedings of the Annual Symposium on Foundations of
  Computer Science (FOCS)}, pages 454--462, 2004.

\bibitem{lee2017improved}
Euiwoong Lee, Melanie Schmidt, and John Wright.
\newblock Improved and simplified inapproximability for k-means.
\newblock {\em Information Processing Letters}, 120:40--43, 2017.

\bibitem{makarychev2016bi}
Konstantin Makarychev, Yury Makarychev, Maxim Sviridenko, and Justin Ward.
\newblock A bi-criteria approximation algorithm for $ k $ means.
\newblock In {\em Proceedings of the International Workshop on Approximation,
  Randomization, and Combinatorial Optimization Algorithms and Techniques
  (APPROX-RANDOM)}, 2016.

\bibitem{makarychev2014bilu}
Konstantin Makarychev, Yury Makarychev, and Aravindan Vijayaraghavan.
\newblock Bilu-linial stable instances of max cut and minimum multiway cut.
\newblock In {\em Proceedings of the Annual Symposium on Discrete Algorithms
  (SODA)}, pages 890--906, 2014.

\bibitem{mihalak2011complexity}
Mat{\'u}{\v{s}} Mihal{\'a}k, Marcel Sch{\"o}ngens, Rastislav {\v{S}}r{\'a}mek,
  and Peter Widmayer.
\newblock On the complexity of the metric tsp under stability considerations.
\newblock In {\em SOFSEM: Theory and Practice of Computer Science}, pages
  382--393. Springer, 2011.

\bibitem{ostrovsky2012effectiveness}
Rafail Ostrovsky, Yuval Rabani, Leonard~J Schulman, and Chaitanya Swamy.
\newblock The effectiveness of lloyd-type methods for the k-means problem.
\newblock {\em Journal of the ACM (JACM)}, 59(6):28, 2012.

\bibitem{tim}
Tim Roughgarden.
\newblock Beyond worst-case analysis.
\newblock http://theory.stanford.edu/~tim/f14/f14.html, 2014.

\bibitem{vishwanathan}
Sundar Vishwanathan.
\newblock An o(log*n) approximation algorithm for the asymmetric p-center
  problem.
\newblock In {\em Proceedings of the Annual Symposium on Discrete Algorithms
  (SODA)}, 1996.

\bibitem{voevodski2011min}
Konstantin Voevodski, Maria-Florina Balcan, Heiko R{\"o}glin, Shang-Hua Teng,
  and Yu~Xia.
\newblock Min-sum clustering of protein sequences with limited distance
  information.
\newblock In {\em International Workshop on Similarity-Based Pattern
  Recognition}, pages 192--206, 2011.

\end{thebibliography}
\bibliographystyle{plain}

\appendix


\section{Prior algorithms in the context of local stability} \label{app:algos}

In this section, we discuss previous algorithms in the context of our new local stability framework.
All of these results are useful under the standard definition of perturbation resilience for which they were designed.
However, we will see that the prior algorithms (except for the one designed specifically for $k$-center) use the \emph{global} structure of
the data, which causes the algorithms to behave poorly when a fraction of the dataset is not perturbation resilient.

We start with the recent algorithm of Angelidakis et al.~\cite{angelidakis2017algorithms} 
to optimally cluster 2-perturbation resilient instances for any center-based objective
(which includes $k$-median, $k$-means, and $k$-center). The algorithm is intuitively simple to describe.
The first step is to create a minimum spanning tree $\mathcal{T}$ on the dataset. 
The second step is to perform dynamic programming on $\mathcal{T}$ to find the
$k$-clustering with the lowest cost. The key fact is that 2-perturbation resilience implies that each cluster $C_i$ is a connected subtree in $\mathcal{T}$.
This fact is partially preserved for each 2-LPR cluster $C_i$. For example, it can be made to hold if nearby clusters are also 2-LPR. 
However, the dynamic programming step heavily relies on the entire dataset satisfying
stability. For instance, consider a dataset in which there are $k$ clusters in a line, and all clusters are 2-LPR except for the cluster in the center.
It is possible for the non-LPR cluster causes an offset in the dynamic program, which forces the minimum cost pruning to split each 2-LPR cluster in two,
and merge each half with half of its neighboring cluster.
Therefore, none of the LPR clusters are returned by the algorithm.
We conclude that the dynamic programming step is not robust with respect to perturbation resilience.

Next, we consider the algorithm of Balcan and Liang~\cite{balcan2012clustering} 
to optimally cluster $(1+\sqrt{2})$-perturbation resilient instances for center-based objectives. This was the leading algorithm for perturbation resilient instances
until the recent result by Angelidakis et al.
This algorithm also utilizes a dynamic programming step, although on a different type of tree. 
The algorithm starts with a linkage procedure, which the authors call closure linkage. It starts with $n$ singleton sets, and merges the two sets with the minimum
\emph{closure distance}, which is the minimum radius that covers the sets and satisfies a margin condition that exploits the perturbation resilient structure.
The merge procedure is iterated until all sets merge into a single set containing all $n$ points.
Thus, the algorithm creates a tree $\mathcal{T}$ where the leaves are the $n$ singleton sets, each internal node is a set of points, 
and the root is the set of all $n$ datapoints.
If the dataset satisfies $(1+\sqrt{2})$-perturbation resilience, then each optimal cluster appears as a node in the tree.
Then dynamic programming on $\mathcal{T}$ to find the minimum $k$-pruning outputs the optimal clusters.
Just as in the previous algorithm, the initial guarantee is partially satisfied for LPR clusters.
For example, if a cluster $C_i$ and all nearby clusters satisfy $(1+\sqrt{2})$-LPR, then $C_i$ will appear as a node in $\mathcal{T}$.
However, the dynamic programming step is again not robust to just a few non-LPR clusters.
For example, consider a dataset where $k/2$ clusters are LPR, $k/2$ clusters are non-LPR.
It is possible that the minimum closure distance for each non-LPR cluster contains all $k/2$ non-LPR clusters.
Then the non-LPR clusters are grouped together in a single merge step, forcing the dynamic program to split every LPR cluster in half so that $k$
clusters are outputted. In this example, the LPR clusters can be very far away from the non-LPR clusters, but the algorithm still fails.

Finally, we consider the algorithm of Balcan et al.~\cite{symmetric} which returns the optimal asymmetric $k$-center solution under 2-perturbation resilience.
This algorithm starts by finding a subset of the datapoints which ``behave symmetrically.'' We note that their symmetric set is equivalent to the set of all CCVs.
Next, the algorithm uses a margin condition to separate out the optimal clusters in the symmetric set, and greedily attaches the non-symmetric points at the end.
The first part of the algorithm will correctly output all LPR clusters within the symmetric set. However, when the non-symmetric points are reattached,
many points from non-LPR clusters can attach to LPR clusters, forcing the final outputted clusters to have a huge radius.
There is no easy fix to this algorithm, since the LPR clusters might ``steal'' the centers of non-LPR clusters, 
so the optimal radius of the remaining points may increase significantly.

In conclusion, existing algorithms for $k$-median and $k$-means heavily rely on the global structure provided by perturbation resilience.
If just a single cluster does not satisfy local perturbation resilience, then the global structure is broken and the algorithm may not output any of the optimal
clusters. Therefore, algorithms which exploit the \emph{local} structure of perturbation resilience, such as in Section~\ref{sec:approx}, are needed to have
guarantees that are robust with respect to the level of perturbation resilience of the dataset.
For $k$-center, the nature of the objective ensures that all algorithms have local guarantees, in some sense. However, it is still nontrivial
to exploit the structure of the clusters satisfying local perturbation resilience, while ensuring the guarantees for the rest of the dataset are not
arbitrarily bad (which is what we accomplish in Section~\ref{sec:akc}).

\section{Proofs from Section~\ref{sec:prelim}} \label{app:prelim}

In this section, we give a proof of Lemma~\ref{lem:eps-iff}.

\medskip
\noindent \textbf{Lemma~\ref{lem:eps-iff} (restated).}
\emph{
A clustering instance $(S,d)$ satisfies $(\alpha,\epsilon)$-PR if and only if
each optimal cluster $C_i$ satisfies $(\alpha,\epsilon_i)$-LPR and $\sum_i\epsilon_i\leq2\epsilon n$.
}

\begin{proof}
Given a clustering instance $(S,d)$ satisfying $(\alpha,\epsilon)$-PR, given an $\alpha$-perturbation $d'$ and optimal clustering
$C'_1,\dots,C'_i$ under $d'$, then there exists $\sigma$ such that $\sum_{i=1}^k |C_i\setminus C_{\sigma(i)}'|\leq\epsilon n$.
WLOG, we let $\sigma$ equal the identity permutation.
Now we claim that $\sum_{i=1}^k |C_i\setminus C'_i|=\sum_{i=1}^k |C'_i\setminus C_i|$.
Intuitively, this is true because $C_1,\dots,C_k$ and $C'_1,\dots,C'_k$ are both partitions of the same point set $k$.
Formally,
\begin{align*}
\sum_{i=1}^k |C_i\setminus C'_i|&=\sum_{i=1}^k |(C_i\cup C'_i)\setminus C'_i|&&\text{because }C'_i\subseteq C'_i\\
&=\sum_{i=1}^k \left(|C_i\cup C'_i|-|C'_i|\right)&&\text{because }C'_i\subseteq (C_i\cup C'_i)\\
&=\sum_{i=1}^k |C_i\cup C'_i|-\sum_{i=1}^k |C'_i|&&\\
&=\sum_{i=1}^k |C_i\cup C'_i|-\sum_{i=1}^k |C_i|&&\text{because }\sum_{i=1}^k |C'_i|=\sum_{i=1}^k |C_i|\\
&=\sum_{i=1}^k \left(|C_i\cup C'_i|-|C_i|\right)&&\\
&=\sum_{i=1}^k |(C_i\cup C'_i)\setminus C_i|&&\text{because }C_i\subseteq (C_i\cup C'_i)\\
&=\sum_{i=1}^k |C'_i\setminus C_i|&&\text{because }C_i\subseteq C_i\\
\end{align*}

Now for each $i$, set $\epsilon_i=|C_i\setminus C_i'|+|C'_i\setminus C_i|.$
Clearly, this ensures $C_i$ satisfies $(\alpha,\epsilon_i)$-LPR.
Finally, we have
\begin{align*}
\sum_i\epsilon_i&=\sum_i \left(|C_i\setminus C_i'|+|C'_i\setminus C_i|\right)\\
&=2\sum_i |C_i\setminus C_i'|\\
&\leq 2\epsilon n.
\end{align*}

Now we prove the reverse direction. 
Given a clustering instance $(S,d)$ such that each cluster $C_i$ satisfies $(\alpha,\epsilon_i)$-LPR
and $\sum_i\epsilon_i\leq 2\epsilon n$, given an $\alpha$-perturbation $d'$ and optimal clustering
$C'_1,\dots,C'_i$ under $d'$, then for each $i$, $|C_i\setminus C_i'|+|C'_i\setminus C_i|\leq\epsilon_i.$
Therefore,
\begin{align*}
\sum_i |C_i\setminus C_i'|&=
\frac{1}{2}\sum_i \left(|C_i\setminus C_i'|+|C'_i\setminus C_i|\right)\\
&\leq\epsilon n.
\end{align*}
This concludes the proof.
\end{proof}

\section{Proofs from Section~\ref{sec:approx}} \label{app:approx}

In this section, we give the details from Section~\ref{sec:approx}.

\medskip
\noindent \textbf{Theorem~\ref{thm:approx} (restated).}
\emph{
Given a $k$-median instance $(S,d)$, running local search with search size $\frac{1}{\epsilon}$ returns a clustering that contains
every $(3+2\epsilon)$-LPR cluster, and it gives a $(3+2\epsilon)$-approximation overall.
}

\begin{proof}
Given a $k$-median instance $(S,d)$, let $X$ denote a set of centers obtained by running local search with search size $\frac{1}{\epsilon}$.
Let $Y$ denote a different set of $k$ centers.
As in the proof of Lemma~\ref{lem:approx}, let $A_1,A_2,A_3$, and $A_4$ denote 
$\text{Vor}_X(X\cap Y)\cap\text{Vor}_Y(X\cap Y)$, 
$\text{Vor}_X(X\cap Y)\setminus\text{Vor}_Y(X\cap Y)$,
$\text{Vor}_Y(X\cap Y)\setminus\text{Vor}_X(X\cap Y)$, and 
$\text{Vor}_X(X\setminus Y)\cap \text{Vor}_Y(Y\setminus X)$, respectively (see Figure~\ref{fig:approx}).
From Cohen-Addad and Schwiegelshohn~\cite{cohen2017one}, from Lemmas B.3 and B.4, we have the following.
\begin{align*}
\sum_{v\in A_2\cup A_4}d(v,X)\leq\sum_{v\in A_2\cup A_4}d(v,Y)+2(1+\epsilon)\sum_{v\in A_3\cup A_4}d(v,Y)
\end{align*}
Given a point $v\in A_2$, then its center in $X$ is from $X\cap Y$, and its center from $Y$ is in $Y\setminus X$.
We deduce that $d(v,Y)\leq d(v,X)$, otherwise $v$'s center from $Y$ would be the same as in $X$.
Similarly, for a point $v\in A_3$, we can conclude that $d(v,X)\leq d(v,Y)$, by definition of $A_3$.

Therefore,
\begin{align*}
\sum_{v\in A_2\cup A_3\cup A_4}d(v,X)&\leq\sum_{v\in A_3}d(v,X)+\sum_{v\in A_2\cup A_4}d(v,X)\\
&\leq\sum_{v\in A_3}d(v,Y)+\left(\sum_{v\in A_2\cup A_4}d(v,Y)+2(1+\epsilon)\sum_{v\in A_3\cup A_4}d(v,Y)\right)\\
&\leq\sum_{v\in A_2}d(v,Y)+(3+2\epsilon)\sum_{v\in A_3\cup A_4}d(v,Y)\\
&\leq\sum_{v\in A_2}\min(d(v,X),(3+2\epsilon)d(v,Y))+(3+2\epsilon)\sum_{v\in A_3\cup A_4}d(v,Y)\\
\end{align*}
Now the proof follows from Lemma~\ref{lem:approx}.
\end{proof}

\medskip
\noindent \textbf{Lemma~\ref{lem:d'_metric} (restated).}
\emph{
Given $\alpha\geq 1$ and an 
asymmetric $k$-center clustering instance $(S,d)$ with optimal radius $r^*$,
let $d''$ denote an $\alpha$-perturbation such that for all $u,v$, either
$d''(u,v)=\min(\alpha r^*,\alpha d(u,v))$ or $d''(u,v)=\alpha d(u,v)$.
Let $d'$ denote the metric completion of $d''$.
Then $d'$ is an $\alpha$-metric perturbation of $d$, and the optimal cost under $d'$ is $\alpha r^*$.
}

\begin{proof}
By construction, $d'(u,v)\leq d''(u,v)\leq\alpha d(u,v)$.
Since $d$ satisfies the triangle inequality, we have that $d(u,v)\leq d'(u,v)$,
so $d'$ is a valid $\alpha$-metric perturbation of $d$.

Now given $u,v$ such that $d(u,v)\geq r^*$, we will prove that $d'(u,v)\geq \alpha r^*$.
By construction, $d''(u,v)\geq\alpha r^*$.
Then since $d'$ is the metric completion of $d''$, there exists a path 
$u=u_0$--$u_1$--$\cdots$--$u_{s-1}$--$u_s=v$
such that $d'(u,v)=\sum_{i=0}^{s-1}d'(u_i,u_{i+1})$ and
for all $0\leq i\leq s-1$, $d'(u_i,u_{i+1})=d''(u_i,u_{i+1})$.
If there exists an $i$ such that $d''(u_i,u_{i+1})\geq\alpha r^*$, then $d'(u,v)\geq\alpha r^*$ and we are done.
Now assume for all $0\leq i\leq s-1$, $d''(u_i,u_{i+1})<\alpha r^*$.
Then by construction, $d'(u_i,u_{i+1})=d''(u_i,u_{i+1})=\alpha d(u_i,u_{i+1})$,
and so $d'(u,v)=\sum_{i=0}^{s-1}d'(u_i,u_{i+1})=\alpha \sum_{i=0}^{s-1}d(u_i,u_{i+1})\geq
\alpha d(u,v)\geq \alpha r^*$.

We have proven that for all $u,v$, if $d(u,v)\geq r^*$, then $d'(u,v)\geq\alpha r^*$.
Assume there exists a set of centers $C'=\{c_1',\dots, c_k'\}$ whose $k$-center cost under $d'$
is $<\alpha r^*$. Then for all $i$ and $s\in\text{Vor}_{C',d'}(c_i')$, $d'(c_i',s)<r^*$,
implying $d(c_i',s)<\alpha r^*$ by construction. It follows that the $k$-center cost of
$C'$ under $d$ is $r^*$, which is a contradiction.
Therefore, the optimal cost under $d'$ must be $\alpha r^*$.
\end{proof}

\medskip
\noindent \textbf{Theorem~\ref{thm:kcenter} (restated).}
\emph{
Given an asymmetric $k$-center clustering instance $(S,d)$ and an $\alpha$-approximate clustering $\mathcal{C}$,
each $\alpha$-LPR cluster is contained in $\mathcal{C}$, even under the weaker metric perturbation resilience condition.
}

\begin{proof}
Given an $\alpha$-approximate solution $\mathcal{C}$ to a clustering instance $(S,d)$, and given
an $\alpha$-LPR cluster $C_i$, we will create an $\alpha$-perturbation
as follows. For all $v\in S$, set $d''(v,\mathcal{C}(v))=\min\{\alpha r^*,\alpha d(v,\mathcal{C}(v))\}$.
For all other points $u\in S$, set $d''(v,u)=\alpha d(v,u)$.
Then by Lemma \ref{lem:d'_metric}, the metric completion $d'$ of $d''$ is an $\alpha$-perturbation of $d$
with optimal cost $\alpha r^*$.
By construction, the cost of $\mathcal{C}$ is $\leq \alpha r^*$ under $d'$, therefore, $\mathcal{C}$ is an optimal clustering.
Denote the set of centers of $\mathcal{C}$ by $C$.
By definition of $\alpha$-LPR, there exists $v_i\in C$ such that
$\text{Vor}_C(v_i)=C_i$ in $d'$. Now, given $v\in C_i$, $\text{argmin}_{u\in C}d'(u,v)=v_i$,
so by construction, $\text{argmin}_{u\in C}d(u,v)=v_i$. Therefore, $\text{Vor}_C(v_i)=C_i$, so $C_i\in\mathcal{C}$.
\end{proof}

\section{Proofs from Section \ref{sec:akc}} \label{app:akc}

In this section, we give the details of the proofs from Section~\ref{sec:akc}.

\medskip
\noindent \textbf{Lemma~\ref{lem:properties} (restated).}
\emph{
Given an asymmetric $k$-center clustering instance $(S,d)$ and a 2-LPR cluster $C_i$,
$c_i$ satisfies CCV-proximity and center-separation.
Furthermore, given a CCV $c\in C_i$, a CCV $c'\notin C_i$, and a point $v\in C_i$, we have
$d(c,v)<d(c',v)$.
}

\begin{proof}
Given an instance $(S,d)$ and a 2-metric perturbation resilient cluster $C_i$,
we show that $c_i$ has the desired properties.

\emph{Center Separation:}
Assume there exists a point $v\in C_j$ for $j\neq i$ such that $d(v,c_i)\leq r^*$.
The idea is to construct a $2$-perturbation in which $v$ becomes the center for $C_i$.

\begin{equation*}
d''(s,t)=
\begin{cases}
\min( 2 r^*,2 d(s,t) ) & \text{if } s=v \text{, } t\in C_i  \\
2 d(s,t) & \text{otherwise.}
\end{cases}
\end{equation*}
$d''$ is a valid 2-perturbation of $d$ because for each point $u\in C_i$, 
$d(v,u)\leq d(v,c_i)+d(c_i,u)\leq 2r^*$.
Define $d'$ as the metric completion of $d''$.
Then by Lemma \ref{lem:d'_metric}, $d'$ is a 2-metric perturbation with
optimal cost $2 r^*$.
The set of centers $\{c_{i'}\}_{i'=1}^k\setminus\{c_i\}\cup\{v\}$ achieves the optimal
cost, since $v$ is distance $2 r^*$ from $C_i$, and all other 
clusters have the same center as in $\mathcal{OPT}$ (achieving radius $2 r^*$).
But in this new optimal clustering, $c_i$'s center is a point in $\{c_{i'}\}_{i'=1}^k\setminus\{c_i\}\cup\{v\}$,
none of which are from $C_i$. We conclude that $C_i$ is no longer an optimal cluster, contradicting 2-LPR.

\emph{Final property}
Next we prove the final part of the lemma.
Given a CCV $c$ from a 2-LPR cluster $C_i$, a CCV $c'\in C_j$ such that $j\neq i$, and a point $v\in C_i$,
and assume $d(c',v)\leq d(c,v)$.
We will construct a perturbation in which $c$ and $c'$ become centers of their respective
clusters, and then $v$ switches clusters.
Define the following perturbation $d''$.

\begin{equation*}
d''(s,t)=
\begin{cases}
\min( 2 r^*, 2 d(s,t) ) & \text{if } s=c \text{, } t\in C_i\text{ or }s=c'\text{, }t\in
C_j\cup\{v\} \\
2 d(s,t) & \text{otherwise.}
\end{cases}
\end{equation*}
$d''$ is a valid 2-perturbation of $d$ because for each point $u\in C_i$,
$d(c,u)\leq d(c,c_i)+d(c_i,u)\leq 2r^*$, for each point $u\in C_j$,
$d(c',u)\leq d(c',c_j)+d(c_j,u)\leq 2r^*$, and $d(c',v)\leq d(c,v)\leq d(c,c_i)+d(c_i,v)\leq 2r^*$.
Define $d'$ as the metric completion of $d''$.
Then by Lemma \ref{lem:d'_metric}, $d'$ is a 2-metric perturbation with
optimal cost $2 r^*$.
The set of centers $\{c_{i'}\}_{i'=1}^k\setminus\{c_i,c_j\}\cup\{c,c'\}$ achieves the optimal
cost, since $c$ and $c'$ are distance $2r^*$ from $C_i$ and $C_j$, and all other 
clusters have the same center as in $\mathcal{OPT}$ (achieving radius $2 r^*$).
Then since $d'(c',v)<d(c,v)$, $v$ will not be in the same optimal cluster as $c_i$, causing
a contradiction.

\emph{CCV-proximity:}
By center-separation, we have that $\Gamma^-(c_i)\subseteq C_i$, and by definition of $r^*$,
we have that $C_i\subseteq\Gamma^+(c_i)$. Therefore, $\Gamma^-(c_i)\subseteq\Gamma^+(c_i)$,
so $c_i$ is a CCV
Now given a point $v\in \Gamma^-(c_i)$ and a CCV $c\notin \Gamma^+(c_i)$,
from center-separation and definition of $r^*$,
$v\in C_i$ and $c\in C_j$ for $j\neq i$.
Then from the property in the previous paragraph, 
$d(c_i,v)<d(c,v)$.
\end{proof}

\medskip
\noindent \textbf{Theorem~\ref{thm:asy_local} (restated).}
\emph{
Given an asymmetric $k$-center clustering instance $(S,d)$ of size $n$, 
Algorithm \ref{alg:asy-pr} returns each 2-SLPR cluster exactly.
For each 2-LPR cluster $C_i$, Algorithm \ref{alg:asy-pr} outputs a cluster that is a superset of $C_i$
and does not contain any other 2-LPR cluster.
These statements hold for metric perturbation resilience as well.
Finally, the overall clustering returned by Algorithm \ref{alg:asy-pr} is an $O(\log^* n)$-approximation.
}

\begin{proof}[Proof of Theorem~\ref{thm:asy_local}]
First we explain why Algorithm \ref{alg:asy-pr} retains the approximation guarantee of Algorithm \ref{alg:asy}.
Given any CCV $c\in C_i$ chosen in Phase I, since $c$ is a CCV, then $c_i\in\Gamma^+(c)$, and by definition
of $r^*$, $C_i\subseteq\Gamma^+(c_i)$. Therefore, each chosen CCV always marks its cluster, and we start
Phase II with no remaining CCVs. 
This condition is sufficient for Phase II to return an
$O(\log^* n)$ approximation (Theorem 3.1 from~\cite{vishwanathan}).

Next we claim that for each 2-LPR cluster $C_i$, there exists a cluster outputted by Algorithm~\ref{alg:asy-pr} 
that is a superset of $C_i$ and does not contain any other 2-LPR cluster.
To prove this claim, we first show there exists a point from $C_i$ satisfying CCV-proximity that cannot be marked by any point from
a different cluster in Phase I.
From Lemma~\ref{lem:properties}, $c_i$ satisfies CCV-proximity and
center-separation. If a point $c\notin C_i$ marks $c_i$, then
$\exists v\in\Gamma^-(c)\cap\Gamma^-(c_i)$.
By center-separation and by definition of CCV, 
$c\notin\Gamma^-(c_i)$ and $c_i\notin\Gamma^(c)$.
Then from the definition of CCV-proximity for $c_i$ and $c$,
we have $d(c,v)<d(c_i,v)$ and $d(c_i,v)<d(c,v)$, so we have reached
a contradiction.
At this point, we know a point $c\in C_i$ will always be chosen by the algorithm in Phase I.
To finish the claim, we show that each point $v$ from $C_i$ is closer to $c$ than to any other point $c'\notin C_i$ chosen in Phase I.
Since $c$ and $c'$ are both CCVs, this follows directly from 
the last property in Lemma~\ref{lem:properties}.
However, it is possible that a center $c'\in A_{i+1}$ is closer to $v$ than $c$ is to $v$, causing $c'$ to ``steal'' $v$; this is unavoidable.
Therefore, we forbid the algorithm from decreasing the size of the Voronoi tiles of $C$ after Phase I.

Finally, we claim that Algorithm \ref{alg:asy-pr} returns each 2-SLPR cluster exactly.
Given a 2-SLPR cluster $C_i$, by our previous argument,
there exists a CCV $c\in C_i$ from Phase I satisfying CCV-proximity such that $C_i\subseteq V_c$.
First we assume towards contradiction that there exists a point $v\in\Gamma^-(c)\setminus C_i$.
Let $v\in C_j$. Since $c$ is a CCV, then $v\in\Gamma^+(c)$, so $C_j$ must be 2-LPR by definition of 2-SLPR.
By Lemma \ref{lem:properties}, $c_j$ is a CCV and $d(c_j,v)<d(c,v)$. But this violates CCV-proximity of $c$, so we have
reached a contradiction. Therefore, $\Gamma^-(c)\subseteq C_i$. 
Now assume there exists $v\in V_c\setminus C_i$ at the end of the algorithm.

Case 1: $v$ was marked by $c$ in Phase I. Let $v\in C_j$.
Then there exists a point $u\in \Gamma^-(c)$ such that $v\in\Gamma^+(u)$.
Then $u\in C_i$ and $v\in\Gamma^+(u)$, so $C_j$ must be 2-LPR.
Since $v$ is from a different 2-LPR cluster, it cannot be contained in $V_c$, so we have a contradiction.

Case 2: $v$ was not marked by $c$ in Phase I.
Denote the shortest path in $D_{(S,d)}$ from $c$ to $v$ by $c=v_0$--$v_1$--$\cdots$--$v_{L-1}$--$v_L=v$.
Let $v_\ell\in C_j$ denote the first vertex on the shortest path that is not in $C_i$ (such a vertex must exist because $v\notin C_i$).
See Figure~\ref{fig:slpr}.
Then $v_{\ell-1}\in C_i$ and $d(v_{\ell-1},v_\ell)\leq r^*$, so $C_j$ is 2-LPR.
Let $c'$ denote the CCV chosen in Phase I such that $C_j\subseteq V_{c'}$.
If $v_\ell$ is not on the shortest path $c'$--$v$, then that shortest path must be shorter than $c$--$v$, and we are done.
If $v_{\ell-1}$ is on the shortest path $c'$--$v$, then since $v_\ell\in C_j$, 
$d(c',v_\ell)\leq 2r^*$, so the shortest path must start with $c'$--$v_{\ell-1}$--$v_\ell$.
If $d(c,v_{\ell-1})>r^*$, then we are done because $c'$ is then closer to $v$.
Therefore, the distance from both $c$ and $c'$ to $v_{\ell}$ is in $(r^*,2r^*]$.
We can set up a 2-perturbation as in the proof of the final property of Lemma~\ref{lem:properties},
so that $c$ and $c'$ become the centers of their respective clusters, and $v_{\ell}$ switches clusters.

\begin{equation*}
d''(s,t)=
\begin{cases}
\min( 2 r^*, 2 d(s,t) ) & \text{if } s=c \text{, } t\in C_i\text{ or }s=c'\text{, }t\in
C_j\cup\{v_\ell\} \\
2 d(s,t) & \text{otherwise.}
\end{cases}
\end{equation*}
$d''$ is a valid 2-perturbation of $d$ because for each point $u\in C_i$,
$d(c,u)\leq d(c,c_i)+d(c_i,u)\leq 2r^*$, for each point $u\in C_j$,
$d(c',u)\leq d(c',c_j)+d(c_j,u)\leq 2r^*$, and $d(c',v_\ell)\leq d(c,v_\ell)\leq d(c,c_i)+d(c_i,v_\ell)\leq 2r^*$.
Define $d'$ as the metric completion of $d''$.
Then by Lemma \ref{lem:d'_metric}, $d'$ is a 2-metric perturbation with
optimal cost $2 r^*$.
The set of centers $\{c_{i'}\}_{i'=1}^k\setminus\{c_i,c_j\}\cup\{c,c'\}$ achieves the optimal
cost, since $c$ and $c'$ are distance $2r^*$ from $C_i$ and $C_j$, and all other 
clusters have the same center as in $\mathcal{OPT}$ (achieving radius $2 r^*$).
Then since $d(c,v_\ell)$ and $d(c',v_\ell)$ are both in $(r^*,2r^*]$, by construction
$d'(c,v_\ell)=d'(c',v_\ell)$. This contradicts 2-LPR, since $v_\ell$ can switch clusters to $C_i$.

We conclude that $v_\ell$ is the first common vertex on the shortest paths $c$--$v$ and $c'$--$v$.
Since $c$ and $c'$ are both CCVs, $d(c',v_\ell)<d(c,v_\ell)$. Therefore, $v$ cannot be in $V_c$, so we have reached a contradiction.
This completes the proof.
\end{proof}

\section{Proofs from Section \ref{sec:3eps}} \label{app:3eps}

In this section, we give the details of the proofs from Section~\ref{thm:3epsthm}.

\medskip
\noindent \textbf{Lemma~\ref{lem:closepoints} (restated).}
\emph{
Given a $k$-center clustering instance $(S,d)$
such that all optimal clusters are size $>2\epsilon n$
and there exist two points at distance $r^*$
from different $(3,\epsilon)$-LPR clusters,
then there exists a partition $S_x\cup S_y$ of the non-centers $S\setminus\{c_\ell\}_{\ell=1}^k$ 
such that for all pairs $p\in S_x$, $q\in S_y$, $\{c_\ell\}_{\ell=1}^k\cup\{p,q\}$ $(3,3)$-hits $S$.
}

\begin{proof}
This proof is split into two main cases. The first case is the following: there exists a CCC2 for a
$(3,\epsilon)$-LPR cluster, discounting a $(3,\epsilon)$-LPR cluster.
In fact, in this case, we do not need the assumption that two points from different LPR clusters are close.
If there exists a CCC to a $(3,\epsilon)$-LPR cluster,
denote the CCC by $c_x$ and the cluster by $C_y$. Otherwise, let $c_x$ denote a CCC2
to a $(3,\epsilon)$-LPR cluster $C_y$, discounting a $(3,\epsilon)$-LPR center $c_z$.
Then $c_x$ is at distance $\leq r^*$ to all but
$\epsilon n$ points in $C_y$. Therefore, $d(c_x,c_y)\leq 2r^*$ and so $c_x$ is at
distance $\leq 3r^*$ to all points in $C_y$. Consider the following perturbation $d''$.

\begin{equation*}
d''(s,t)=
\begin{cases}
\min( 3 r^*, 3 d(s,t) ) & \text{if } s=c_x \text{, } t\in C_y  \\
3 d(s,t) & \text{otherwise.}
\end{cases}
\end{equation*}

This is a $3$-perturbation because for all $v\in C_y$, $d(c_x,v)\leq 3 r^*$.
Define $d'$ as the metric completion of $d''$.
Then by Lemma \ref{lem:d'_metric}, $d'$ is a 3-metric perturbation with
optimal cost $3 r^*$.
Given any non-center $v\in S$,
the set of centers $\{c_\ell\}_{\ell=1}^k\setminus\{c_y\}\cup\{v\}$ achieves the optimal
score, since $c_x$ is at distance $3 r^*$ from $C_y$, and all other 
clusters have the same center as in $\mathcal{OPT}$ (achieving radius $3 r^*$). 
Therefore, from Fact~\ref{fact:half}, one of the centers in 
$\{c_\ell\}_{\ell=1}^k\setminus\{c_y\}\cup\{v\}$ must be the center for
the majority of points in $C_y$ under $d'$.
If this center is $c_\ell$, $\ell\neq x,y$, 
then for the majority of points $u\in C_y$, 
$d(c_\ell,u)\leq r^*$ and $d(c_\ell,u)< d(c_z,u)$ for all $z\neq \ell,y$.
Then by definition, $c_\ell$ is a CCC for the $(3,\epsilon)$-LPR cluster, $C_y$.
But then by construction, $\ell$ must equal $x$, so we have a contradiction.
Note that if some $c_\ell$ has for the majority of $u\in C_y$,
$d(c_\ell,u)\leq d(c_z,u)$ (non-strict inequality)
for all $z\neq \ell,y$, then there is another equally
good partition in which $c_\ell$ is not the center for the majority of points
in $C_y$, so we still obtain a contradiction.
Therefore, either $v$ or $c_x$ must be the center for the majority of points
in $C_y$ under $d'$.

If $c_x$ is the center for the majority of points in $C_y$, 
then because $C_y$ is $(3,\epsilon)$-LPR, the corresponding cluster must contain
fewer than $\epsilon n$ points from $C_x$. 
Furthermore, since for all $\ell\neq x$ and $u\in C_x$, $d(u,c_x)<d(u,c_\ell)$,
it follows that $v$ must be the center
for the majority of points in $C_x$.
Therefore, every non-center $v\in S$ is at distance $\leq r^*$ to the majority of points in
either $C_x$ or $C_y$.

Now partition all the non-centers into two sets $S_x$ and $S_y$, such that
$S_x=\{p\mid \text{for the majority of points }q\in C_x,\text{ }d(p,q)\leq r^*\}$ and
$S_y=\{p\mid p\notin S_x\text{ and for the majority of points }q\in C_y,\text{ }d(p,q)\leq r^*\}$.
Given $p,q\in S_x$,
there exists an $s\in C_x$ such that
$d(p,q)\leq d(p,s)+d(s,q)\leq 2r^*$ (since both points are close to more
than half of points in $C_x$).
Similarly, any two points $p,q\in S_y$ are $\leq 2r^*$ apart.

For now, assume that $S_x$ and $S_y$ are both nonempty.
Given a pair $p\in S_x$, $q\in S_y$, we claim that $\{c_\ell\}_{\ell=1}^k\cup\{p,q\}$ $(3,3)$-hits $S$.
Given a point $s\in C_i$ such that $i\neq x,y$, WLOG $s\in S_x$.
Then $c_i$, $p$, and $c_x$ are all distance $3r^*$ to $s$.
Furthermore, $c_i$, $c_x$ and $p$ are all distance $3r^*$ to $c_i$.
Given a point $s\in C_x$, then $c_x$, $c_y$, and $p$ are distance $3r^*$ to $s$ because $d(c_x,c_y)\leq 2r^*$.
Finally, $c_x$, $c_y$ and $p$ are distance $3r^*$ to $c_x$, and similar arguments hold for $s\in C_y$ and $c_y$.
Therefore, $\{c_\ell\}_{\ell=1}^k\cup\{p,q\}$ $(3,3)$-hits $S$.

If $S_x=\emptyset$ or $S_y=\emptyset$, then we can prove a slightly stronger statement:
for each pair of non-centers $\{p,q\}$, $\{c_\ell\}_{\ell=1}^k\cup\{p,q\}$ $(3,3)$-hits $S$.
The proof is the same as the previous paragraph.
Thus, the lemma is true when assuming there exists a CCC2 for a
$(3,\epsilon)$-LPR cluster, discounting a $(3,\epsilon)$-LPR cluster.

Now we turn to the other case.
Assume there does not exist a CCC2 to a LPR cluster, discounting a LPR center. 
In this case, we need to use the assumption that there exist $(3,\epsilon)$-LPR clusters $C_x$ and $C_y$,
and $p\in C_x$, $q\in C_y$ such that $d(p,q)\leq r^*$.
Then by the triangle inequality, $p$ is distance $\leq 3r^*$
to all points in $C_x$ and $C_y$.
Consider the following $d''$.

\begin{equation*}
d''(s,t)=
\begin{cases}
\min( 3 r^*, 3 d(s,t) ) & \text{if } s=p \text{, } t\in C_x\cup C_y  \\
3 d(s,t) & \text{otherwise.}
\end{cases}
\end{equation*}

This is a $3$-perturbation because $d(p,C_x\cup C_y)\leq 3 r^*$.
Define $d'$ as the metric completion of $d''$.
Then by Lemma \ref{lem:d'_metric}, $d'$ is a 3-metric perturbation with
optimal cost $3 r^*$.
Given any non-center $s\in S$,
the set of centers $\{c_\ell\}_{\ell=1}^k\setminus\{c_x,c_y\}\cup\{p,s\}$ achieves the optimal
score, since $p$ is distance $3 r^*$ from $C_x\cup C_y$, and all other 
clusters have the same center as in $\mathcal{OPT}$ (achieving radius $3 r^*$). 

From Fact~\ref{fact:half}, one of the centers in 
$\{c_\ell\}_{\ell=1}^k\setminus\{c_x,c_y\}\cup\{p,s\}$ must be the center for
the majority of points in $C_x$ under $d'$.
If this center is $c_\ell$ for $\ell\neq x,y$, 
then for the majority of points $t\in C_x$, 
$d(c_\ell,t)\leq r^*$ and $d(c_\ell,t)< d(c_z,t)$ for all $z\neq \ell,x,y$.
So by definition, $c_\ell$ is a CCC2 for $C_x$ discounting $c_y$, which contradicts our assumption.
Similar logic applies to the center for the majority of points in $C_y$.
Therefore, $p$ and $s$ must be the centers for $C_x$ and $C_y$.
Since $s$ was an arbitrary non-center, all non-centers are 
distance $\leq r^*$ to all but $\epsilon n$ points in either $C_x$ or $C_y$.

Similar to Case 1, we now partition all the non-centers into two sets $S_x$ and $S_y$, such that
$S_x=\{u\mid \text{for the majority of points }v\in C_x,\text{ }d(u,v)\leq r^*\}$ and
$S_y=\{u\mid u\notin S_x\text{ and for the majority of points }v\in C_y,\text{ }d(u,v)\leq r^*\}$.
As before, each pair of points in $S_x$ are distance $\leq 2r^*$ apart, and similarly for $S_y$.
It is no longer true that $d(c_x,c_y)\leq 2r^*$, however, we can prove that for both $S_x$ and $S_y$,
there exist points from two distinct clusters each. From the previous paragraph, given a non-center
$s\in C_i$ for $i\neq x,y$, we know that $p$ and $s$ are centers for $C_x$ and $C_y$.
With an identical argument, given $t\in C_j$ for $j\neq x,y,i$, we can show that $q$ and $t$ are centers for $C_x$ and $C_y$.
It follows that $S_x$ and $S_y$ both contain points from at least two distinct clusters.

Now we finish the proof by showing that for each pair $p\in S_x$, $q\in S_y$, $\{c_\ell\}_{\ell=1}^k\cup\{p,q\}$ $(3,3)$-hits $S$.
Given a non-center $s\in C_i$, WLOG $s\in S_x$, then there exists $j\neq i$ and $t\in C_j\cap S_x$.
Then $c_i$, $c_j$, and $p$ are $3r^*$ to $s$ and $c_i$, $c_x$, and $p$ are $3r^*$ to $c_i$.
In the case where $i=x$, then $c_i$, $c_j$, and $p$ are $3r^*$ to $c_i$.
This concludes the proof.
\end{proof}

\medskip
\noindent \textbf{Fact~\ref{fact:ranking} (restated).}
\emph{
Given a $k$-center clustering instance $(S,d)$
such that all optimal clusters have size $>2\epsilon n$,
and an $\alpha$-perturbation $d'$ of $d$,
let $\mathcal{C'}$ denote the set of 
$(\alpha,\epsilon)$-LPR clusters.
For each $C_x\in\mathcal{C'}$, there exists a bijection $R_{x,d'}:S\rightarrow[n]$
such that for all sets of $k$ centers $C$ that achieve the optimal cost under $d'$, 
then $c=\text{argmin}_{c'\in C} R_{x,d'}(c')$ if and only if $\text{Vor}_C(c)$ is $\epsilon$-close to $C_x$.
}

\begin{proof}
Assume the lemma is false. Then there exists an $(\alpha,\epsilon)$-LPR cluster $C_i$, 
two distinct points $u,v\in S$, and two
sets of $k$ centers $C$ and $C'$ both containing $u$ and $v$, and both sets achieve the optimal score under an $\alpha$-perturbation
$d'$, but $u$ is the center for $C_i$ in $C$ while $v$ is the center for $C_i$ in $C'$.
Then $\text{Vor}_C(u)$ is $\epsilon$-close to $C_i$; similarly, $\text{Vor}_{C'}(v)$ is $\epsilon$-close to $C_i$.
This implies $u$ is closer to all but $\epsilon n$ points in $C_i$ than $v$, and $v$
 is closer to all but $\epsilon n$ points in $C_i$ than $u$.
Since $|C_i|>2\epsilon n$, this causes a contradiction.
\end{proof}

\medskip
\noindent \textbf{Lemma~\ref{lem:rank-struct} (restated).}
\emph{
Given a $k$-center clustering instance $(S,d)$
such that all optimal clusters are size $>2\epsilon n$, 
and given non-centers $p,q\in S$ such that
$C=\{c_\ell\}_{\ell=1}^k\cup\{p,q\}$ $(3,3)$-hits $S$,
let the set $\mathcal{C'}$ denote the set of 
$(3,\epsilon)$-LPR clusters.
Define the 3-perturbation $d'$ as in Lemma \ref{lem:hit}. 
The following are true.
\begin{enumerate}
\item Given $C_x\in \mathcal{C'}$ and $C_i$ such that $i\neq x$, $R_{x,d'}(c_x)<R_{x,d'}(c_i)$.
\item There do not exist $s\in C$ and $C_x,C_y\in\mathcal{C'}$ such that $x\neq y$, and $R_{x,d',C}(s)+R_{y,d',C}(s)\leq 4$.
\item Given $C_i$ and $C_x,C_y\in\mathcal{C'}$ such that $x\neq y\neq i$,
if $R_{x,d',C}(c_i)\leq 3$, then $R_{y,d',C}(p)\geq 3$ and $R_{y,d',C}(q)\geq 3$.
\end{enumerate}
}

\begin{proof}
\begin{enumerate}
\item By definition of the optimal clusters, for each $s\in C_x$, $d(c_x,s)<d(c_i,s)$, and therefore by construction, $d'(c_x,s)<d'(c_i,s)$.
It follows that $R_{x,d'}(c_x)<R_{x,d'}(c_i)$.
\item Assume there exists $s\in C$ and $C_x,C_y\in \mathcal{C'}$ such that $R_{x,d',C}(s)+R_{y,d',C}(s)\leq 4$.

Case 1: $R_{x,d',C}(s)=1$ and $R_{y,d',C}(s)\leq 3$.
Define $u$ and $v$ such that $R_{y,d',C}(u)=1$ and $R_{y,d',C}(v)=2$. 
(If $u$ or $v$ is equal to $s$, then redefine it to an arbitrary center in $C\setminus\{s,u,v\}$.)
Consider the set of centers $C'=C\setminus\{u,v\}$ which is optimal under $d'$ by Lemma \ref{lem:hit}. 
By Fact~\ref{fact:ranking}, $s$ is the center for the majority of points in both $C_x$ and $C_y$, causing a contradiction.

Case 2: $R_{x,d',C}(s)= 2$ and $R_{y,d',C}(s)= 2$.
Define $u$ and $v$ such that $R_{x,d',C}(u)=1$ and $R_{y,d',C}(v)=1$.
(Again, if $u$ or $v$ is equal to $s$, then redefine it to an arbitrary center in $C\setminus\{s,u,v\}$.) 
Consider the set of centers $C'=C\setminus\{u,v\}$ which is optimal under $d'$ by Lemma \ref{lem:hit}. 
However, by Fact~\ref{fact:ranking}, $s$ is the center for the majority of points in both $C_x$ and $C_y$, causing a contradiction.
\item Assume $R_{x,d',C}(c_i)\leq 3$.
 
Case 1: $R_{x,d',C}(c_i)=2$. Then by Lemma~\ref{lem:rank-struct} part 1, 
$R_{x,d',C}(c_x)=1$.
Consider the set of centers $C'=C\setminus\{c_x,p\}$, which is optimal under $d'$.
By Fact \ref{fact:ranking}, $\text{Vor}_{C'}(c_i)$ must be $\epsilon$-close to $C_x$.
In particular, $\text{Vor}_{C'}(c_i)$ cannot contain more than $\epsilon n$ points from $C_i$. 
But by definition, for all $j\neq i$ and $s\in C_i$, $d(c_i,s)<d(c_j,s)$. It follows that
$\text{Vor}_{C'}(q)$ must contain all but $\epsilon n$ points from $C_i$.
Therefore, for all but $\epsilon n$ points $s\in C_i$, for all $j$, $d'(q,s)<d'(c_j,s)$.
If $R_{y,d',C}(q)\leq 2$, then $C_y$ ranks $c_y$ or $p$ number one.
Then for the set of centers $C'=C\setminus\{c_y,p\}$,
$\text{Vor}_{C'}(q)$ contains more than $\epsilon n$
points from $C_y$ and $C_i$, contradicting the fact that $C_y$ is $(3,\epsilon)$-LPR.
Therefore, $R_{y,d',C}(q)\geq 3$. The argument to show $R_{y,d',C}(p)\geq 3$ is identical.

Case 2: $R_{x,d',C}(c_i)=3$. If there exists $j\neq i,x$ such that $R_{x,d',C}(c_i)=2$, then WLOG we
are back in case 1. By Lemma~\ref{lem:rank-struct} part 1, $R_{x,d',C}(c_x)\leq 2$. 
Then either $p$ or $q$ are ranked top two, WLOG $R_{x,d',C}(p)\leq 2$.
Consider the set $C'=C\setminus\{c_x,p\}$.
Then as in the previous case, $\text{Vor}_{C'}(c_i)$ must be $\epsilon$-close to 
$C_x$, implying for all but $\epsilon n$ points $s\in C_i$, for all $j$, $d'(q,s)<d'(c_j,s)$.
If $R_{y,d',C}(q)\leq 2$, again, $C_y$ ranks $c_y$ or $p$ as number one.
Let $C'=C\setminus\{c_y,p\}$, and then $\text{Vor}_{C'}(q)$ contains more than $\epsilon n$
points from $C_y$ and $C_i$, causing a contradiction.
Furthermore, if $R_{y,d',C}(p)\leq 2$, then we arrive at a contradiction by
Lemma~\ref{lem:rank-struct} part 2.
\end{enumerate}
\end{proof}

\end{document}